\documentclass[transmag]{IEEEtran}
	\usepackage{cite}
	\usepackage{amsmath,amssymb,amsfonts}
	\usepackage{algorithm,algorithmicx}
	\usepackage{graphicx}
	\usepackage{textcomp}
	\usepackage{caption}
	\input{def.tex}
	\def\BibTeX{{\rm B\kern-.05em{\sc i\kern-.025em b}\kern-.08em
			T\kern-.1667em\lower.7ex\hbox{E}\kern-.125emX}}
	\begin{document}
		
		\title{Multipair Two-Way DF Relaying with Cell-Free Massive MIMO}
		\author{Anastasios K. Papazafeiropoulos, Pandelis Kourtessis, Symeon Chatzinotas, John M. Senior,
			\thanks{A. Papazafeiropoulos is with the Communications and Intelligent Systems Research Group,
				University of Hertfordshire, Hatfield AL10 9AB, U. K., and with SnT at the University of Luxembourg, Luxembourg. P. Kourtessis and John M. Senior are with the Communications and Intelligent Systems Research Group, University of Hertfordshire, Hatfield AL10 9AB, U. K. S. Chatzinotas is with the SnT at the University of Luxembourg, Luxembourg. E-mails: tapapazaf@gmail.com, p.kourtessis@herts.ac.uk, symeon.chatzinotas@uni.lu.}}
		\IEEEtitleabstractindextext{\begin{abstract}
				We consider a two-way half-duplex decode-and-forward (DF) relaying system with multiple pairs of single-antenna users assisted by a cell-free (CF) massive multiple-input multiple-output (mMIMO) architecture with multiple-antenna access points (APs). Under the practical constraint of imperfect channel state information (CSI), we derive the achievable sum spectral efficiency (SE) for a finite number of APs with maximum ratio (MR) linear processing for both reception and transmission in closed-form. Notably, the proposed CF mMIMO relaying architecture, exploiting the spatial diversity, and providing better coverage, outperforms the conventional collocated mMIMO deployment. Moreover, we shed light on the power-scaling laws maintaining a specific SE as the number of APs grows. A thorough examination of the interplay between the transmit powers per pilot symbol and user/APs takes place, and useful conclusions are extracted. Finally, differently to the common approach for power control in CF mMIMO systems, we design a power allocation scheme maximizing the sum SE.
			\end{abstract}
			
			%
		}
		\maketitle
		
		\section{Introduction}
		\label{Intro}
		Massive multiple-input multiple-output (mMIMO) systems, where a large number of antennas in both collocated and distributed setups serves simultaneously a lower number of users, has become one of the key fifth-generation (5G) physical-layer technologies towards higher throughput and energy efficiency by means of simple linear signal processing \cite{Larsson2014,Marzetta2016}. Recently, a distributed mMIMO architecture under the name of cell-free (CF) mMIMO, enjoying the benefits of network MIMO, has emerged by providing higher coverage probability and exploiting the diversity against shadow fading \cite{Ngo2017}. In particular, a CF mMIMO system includes a large number of single-antenna access points (APs) that is  connected to a central processing unit (CPU) and serves jointly all  users by means of coherent joint transmission/reception. In this direction,  in  \cite{Ngo2018}, APs were equipped with multiple antennas to increase the array and diversity gains. Generally, CF mMIMO systems outperform small cells (SCs) and collocated deployments but
		given that this is an emerging promising architecture at its infancy, its study is limited  \cite{Ngo2018,Buzzi2017a,Bashar2019,Alonzo2019,Bjoernson2020,Papazafeiropoulos2020,Papazafeiropoulos2020a,Papazafeiropoulos2020b}. 
		For example, 
		 authors in~\cite{Buzzi2017a} achieved better data rates by means of a user-centric approach, where the APs serve a group of users instead of all of them, while in \cite{Papazafeiropoulos2020}, the realistic spatial randomness of the APs was taken into account by means of a Poisson point process (PPP) to obtain the coverage probability in CF mMIMO systems.

		Multipair relaying systems, where multiple pairs of users communicate simultaneously by means of a relay to enhance the network coverage, have been improved further by the introduction of the mMIMO characteristics, which enhance the spatial diversity and achieve an order of magnitude spectral efficiency (SE) improvement \cite{Ngo2014,Tan2017}. This technique, known as multipair mMIMO relaying, initially considered one-way transmission by means of amplify-and-forward (AF) as well as decode-and-forward (DF) protocols \cite{Ngo2014,Gao2016,Tan2017,Papazafeiropoulos2018a}. For instance, in the case of DF relaying, the authors in \cite{Ngo2014} examined the achievable SE in Rayleigh fading channels for both maximum-ratio (MR) and zero-forcing (ZF) linear processing, while in \cite{Tan2017}, optimization of the energy efficiency was performed. In reference to AF relaying, the power allocation and max-min user selection was investigated in \cite{Gao2016}. 
		
	Unfortunately, the one-way transmission strategy incurs an SE degradation by $ 50\% $ \cite{Tan2017,Rankov2007}. To reduce this loss, the two-way mMIMO relaying, where bidirectional simultaneous data transmission applies, has attracted significant attention \cite{Cui2014,Feng2017,Dai2016a,Zhang2016,Kong2018,Kong2019,Ngo2013a}. Hence, its main advantage in comparison to one-way relaying is the reduction of the time required  for information exchange between the user pairs into just two time slots. For instance, \cite{Cui2014} and \cite{Feng2017} characterized the power-scaling laws for half-duplex (HD) and full-duplex (FD) transmission, respectively. Still,  first works assumed only perfect channel state information (CSI), which is highly unrealistic since practical systems have the availability of only imperfect CSI.  Thus, a training phase with pilot transmission and minimum mean-square-error (MMSE) estimation at the relay was considered in \cite{Dai2016a} and \cite{Zhang2016} for AF and DF protocols, respectively. Apart from evaluating the impact of imperfect CSI on the system performance, several other important questions arise. For example, \cite{Kong2018} and \cite{Kong2019} investigated the interplay among the transmit powers of the pilot, user, and relay for DF and AF, respectively. Interestingly, distributed relaying for multipair two-way channels has been studied but only for AF and perfect CSI conditions \cite{Wang2011,Ngo2013a}.
				\subsection{Motivation}
		This work relies on vital observations: i) Collocated mMIMO relaying in \cite{Kong2018} has not fully taken advantage of the benefits of practical distributed mMIMO systems. In particular, CF mMIMO systems exploit the spatial diversity, provide better coverage since the APs are closer to the users, and suppress the inter-AP interference due to  the APs cooperation\footnote{Some other differences between the two are architectures follow. First, in CF mMIMO systems, each AP obtains its local CSI. In this way, we achieve to alleviate the computational burden at the CPU. Moreover, CF mMIMO systems achieve increased fairness with respect to the users since, in CF, it is more possible to have an AP close to a user. Also, CF achieves lower latency since the APs are closer to the users. Actually, their advantageous layout is more attractive for mobile edge computing and caching. Notably, in \cite{Bjoernson2020}, being one of the references pointing to the differences between CF mMIMO and single-cell collocated mMIMO systems, the authors highlighted the differences in the generation of the correlation matrices and the allocation of the pilots.}. ii) Two-way relaying is more advantageous than one-way transmission. iii) Given that AF undergoes noise amplification, DF may perform better at a low signal-to-noise ratio (SNR) \cite{Kim2011}. Thus, DF is a more attractive method for mMIMO relaying normally operating at the low power regime \cite{Larsson2014}. iv) DF is more flexible than AF in the case of two-way relaying since it can allow power allocation in both directions \cite{Gao2013}. v) Although some works have studied FD for CF mMIMO systems \cite{Vu2019,Wang2020,Nguyen2019}, FD may not be advisable for CF mMIMO systems due to the large power level difference of the transmit/received signals in the near/far-field while the APs might be quite close to the users. vi) The existing literature on CF mMIMO relaying has not studied at all the insightful power-scaling laws.
		\subsection{Contributions and Outcomes}
		The main contributions are summarized as follows.
		\begin{itemize}
			\item Motivated by the above observations, we   establish the theoretical framework for a two-way CF mMIMO relaying system with imperfect CSI by employing the DF protocol, where a large number of distributed APs plays the role of a relay\footnote{This framework will be the ground to study other interesting topics on CF mMIMO systems and two-way relaying such as the impact of imperfect backhaul links \cite{Bashar2019}.}. Notably, the two-way design enhances the CF mMIMO relaying performance by achieving a significant increase in the SE with comparison to one-way communication but also CF mMIMO is expected to enhance two-way transmission due its accompanying advantages compared to collocated mMIMO systems.
			\item Contrary to the existing work~\cite{Kong2018}, which has studied two-way mMIMO relaying with a large number of collocated antennas, we have accounted for the emerging CF mMIMO architecture at the position of the large collocated deployment. Notably, although~\cite{Kong2018} proposes asymptotic approximations of the SE for a large number of antennas in terms of  deterministic equivalents, we consider a finite number of APs, where the analysis is more general while we also cover the scenario, where the number of APs grows to infinity\footnote{Among other differences between this and~\cite{Kong2018}, we would like to highlight that the  channel model and the power constraints are different. In our work, the power constraints concern each AP and the channel model describes the link from each AP
			to a user. On the contrary, in~\cite{Kong2018}, the power constraints concern the single collocated relay, and the channel model describes the link from the collocated relay to a user. Moreover, our expressions include  summations over the number of APs, corresponding to the contributions from different APs. Such summations do not appear in~\cite{Kong2018} and result in simplified equations while in this work, equations are more complicated and require different manipulations which are relevant to the CF mMIMO analysis. Also, in this work, the channel estimation takes place at each AP. On the contrary, in~\cite{Kong2018}, the
			channel estimation is performed at the collocated relay.}. In addition, contrary to \cite{Ngo2013a}, which assumed the AF protocol and the unrealistic assumption of perfect CSI, we have considered imperfect CSI and the DF protocol, being more suitable for CF mMIMO systems. Also, differently to~\cite{Vu2019,Wang2020,Nguyen2019} which considered FD transmission for CF mMIMO systems, we focus on a totally different architecture being HD two-way transmission, and we investigate the problems of power-scaling laws and power allocation by maximizing the sum SE. Notably, previous works on CF considered only the max-min fairness method to obtain the power control coefficients.
			\item We derive the achievable sum SE of a two-way CF mMIMO relaying network employing the DF protocol with imperfect CSI and linear processing by means of maximum ratio combiner (MRC) for the uplink as well as maximum ratio transmission (MRT) for the downlink in closed form for a finite number of APs and we demonstrate better performance over the collocated deployment with a large number of antennas\footnote{The application of other linear techniques such as zero-forcing presents certain trade-offs between complexity and performance. For instance, they demand more backhaul, which might be prohibitive in the case of large distributed networks (CF mMIMO systems). In particular, the results could also be extended by incorporating the more robust MMSE processing, as suggested recently in \cite{Bjoernson2020}. Given the MMSE intractability, the extension can take place  by simulations or by the deterministic equivalent analysis \cite{Papazafeiropoulos2015a,Papazafeiropoulos2017a,Papazafeiropoulos2020}. The significance of these observations suggests them to be an interesting topic of future research.}. To the best of our knowledge, no other prior work has obtained similar expressions under the CF mMIMO relaying consideration. 
			\item 
We carry out an asymptotic analysis for this architecture to investigate the power-scaling laws maintaining a specific SE as the number of APs increases. Notably, this is the unique work on CF mMIMO relaying that obtains power scaling laws, which are different compared to the scaling laws in~\cite{Kong2018} since they include summations with 	respect to the number of APs and the relevant variables correspond to different APs instead of a  relay with collocated antennas. We observe a trade-off among the transmit power of each pilot symbol, user, and relay, and we shed light on the impact of the scaling parameters.
			\item We formulate an optimization problem for CF mMIMO systems, maximizing the sum SE by keeping constant the total transmit power, in order to obtain the necessary power control coefficients while previous works on CF were relied on a power control by using the max-min fairness approach. Numerical results show the improvement of the sum SE compared to uniform power allocation.
		\end{itemize}
		
		\subsection{Paper Outline} 
		The remainder of this paper is organized as follows. Section~\ref{System} presents the system model of a two-way CF mMIMO relaying system with multiple antennas APs employing the DF protocol. This section includes also the channel estimation and data transmission phases. Section~\ref{SEA} provides the SE analysis for a finite number of APs. Section~\ref{PE} presents the power-scaling laws under different power settings while Section~\ref{PA} addresses the optimal power allocation.
		The numerical results are discussed in Section~\ref{Numerical}, and Section~\ref{Conclusion} concludes the paper.
		
		\subsection{Notation}Vectors and matrices are denoted by boldface lower and upper case symbols, respectively. The symbols $(\cdot)^\T$,  $(\cdot)^\H$, and $(\cdot)^*$ express the transpose, Hermitian transpose, and conjugate operators, respectively. The expectation and variance operators are denoted by $\EE\left[\cdot\right]$ and $\var\left[\cdot\right]$, respectively. Also, $\bb \sim \cC\cN{(\b0,\mathbf{\Sigma})}$ represents a circularly symmetric complex Gaussian vector with {zero mean} and covariance matrix $\mathbf{\Sigma}$.
		
		\section{System Model}\label{System} 
		As illustrated in Fig.~\ref{Fig0}, we consider a CF mMIMO architecture, where a set of APs playing the role of distributed relays, assists the exchange of information in a multipair two-way relaying system. Specifically, we assume that a set $ \mathcal{W}=\{1,\ldots, W\}$ of $ W =|\mathcal{W}|$ communication user pairs, consisted of users $ \mathrm{T}_{\mathrm{A},i} $ and $ \mathrm{T}_{\mathrm{B},i} $, $ i=1,\ldots, W $, is served simultaneously by means of the set $ \mathcal{M}=\{1,\ldots, M\}$ of $ M=|\mathcal{M}| $ APs in the same time-frequency resources. Moreover, all APs connect to a CPU via perfect backhaul links. Certain conditions such as severe shadowing do not allow the existence of direct links between the user pairs. Also, each AP is equipped with $ N $ antennas, and each user has a single antenna. All the nodes, i.e., the user pairs and the APs are randomly distributed in a wide area and operate in the HD mode.

		\begin{figure*}[!h]
			\begin{center}
				\includegraphics[width=0.9\linewidth]{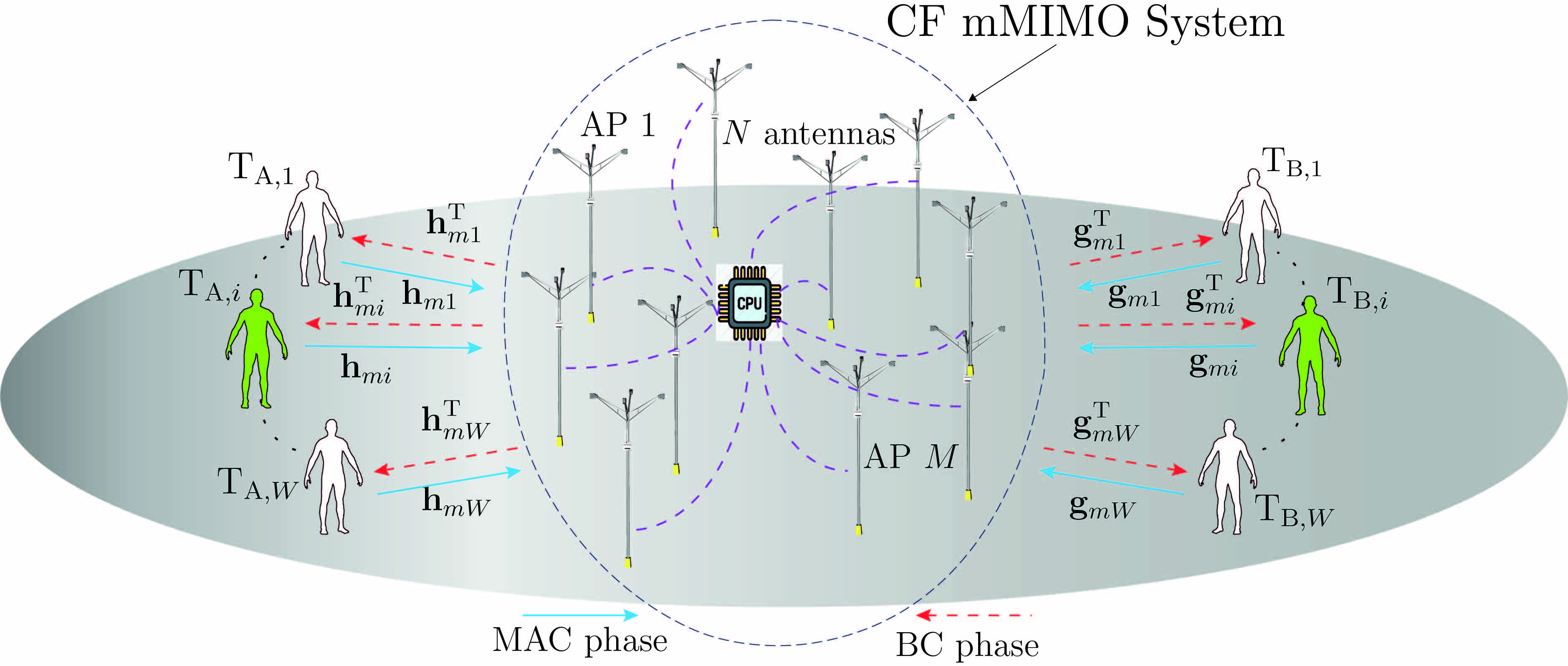}
				\caption{\footnotesize{A multipair two-way CF mMIMO relaying network with $ M $ multi-antenna APs and $ W $ user pairs.}}
				\label{Fig0}
			\end{center}
		\end{figure*}
		
		The system operation takes place within a coherence interval under the time division duplex (TDD) protocol where channel reciprocity is met \cite{Marzetta2016}. The data transmission phase includes two stages, namely, the multiple-access channel (MAC) and broadcasting (BC) stages. In the former, i.e., the MAC stage, all users transmit to the APs. In particular, the message, sent by user $i $, is decoded by means of joint processing at the CPU where all APs have sent their received signals. In a similar way, during the second stage, all APs transmit simultaneously to all user pairs.
		
		The channel model includes both small and large-scale fading. Especially, the large-scale fading describes the effects of shadowing and path-loss. Also, the large-scaling fading changes slowly, i.e., it can be assumed constant for several coherence intervals while the small-scale fading stays static during the duration of a coherence interval but it changes from one interval to the next. In mathematical terms, the channel vector between $ \mathrm{T}_{\mathrm{A},i}$ and $ m $th AP is given by $ \bh_{mi}\in \mathbb{C}^{N \times 1} \sim\mathcal{CN}\left(\b0, \al_{\mathrm{A},mi}\Id_{N}\right) $, known as uncorrelated Rayleigh
		fading, where $ \al_{\mathrm{A},mi} $ represents the large-scale fading of the corresponding link \cite{Ngo2017}. Similarly, the channel between $ \mathrm{T}_{\mathrm{B},i}$ and $ m$th AP is denoted by $ \bg_{mi}\in \mathbb{C}^{N \times 1}\sim\mathcal{CN}\left(\b0, \al_{\mathrm{B},mi}\Id_{N}\right) $ with $ \al_{\mathrm{B},mi} $ being the large-scale fading of this link. 
		
		\subsection{Channel Estimation}
		Given that the propagation channels are piece-wise constant over a coherence interval, both $ \bh_{mi}$ and $ \bg_{mi}$ need to be estimated in every interval by means of pilot transmission. Hence, let $ \tau_{\mathrm{c}} $ and $ \tau_{\mathrm{p}} $ be the durations of the coherence interval and the uplink training in symbols ($ \tau_{\mathrm{p}} <\tau_{\mathrm{c}}$) \cite{Ngo2017}. Both $ \mathrm{T}_{\mathrm{A},i}$ and $ \mathrm{T}_{\mathrm{B},i}$, $ i=1\dots,W $, send simultaneously orthogonal pilot sequences $ \bphi_{\mathrm{A},i}\in \mathbb{C}^{\tau_{\mathrm{p}}\times 1} $ and $ \bphi_{\mathrm{B},i}\in \mathbb{C}^{\tau_{\mathrm{p}}\times 1} $, respectively\footnote{Works on CF  mMIMO systems usually assume non-orthogonal pilots among the users. However, popular works in this area exist  that are based on the assumption of orthogonal pilots e.g., \cite{Nayebi2017}. In this work, which is the first one studying the multipair two-way transmission with CF mMIMO relaying, we have assumed orthogonal pilots among users to enable comparison with existing works on collocated mMIMO systems. Future studies on multipair two-way transmission could include in the design  the use of non-orthogonal pilots.}. Note that this mutual orthogonality requires $\tau_{\mathrm{p}} \ge 2 W $, $ \bphi_{\mathrm{A},i}^{\H}\bphi_{\mathrm{B},i} =0$, and $ \bphi_{\mathrm{A},i}^{\H}\bphi_{\mathrm{A},j}=\bphi_{\mathrm{B},i}^{\H}\bphi_{\mathrm{B},j}=0, ~\forall i\ne j $ \cite{Ngo2014}. In addition, we assume $ \|\bphi_{\mathrm{A},i}\|^{2} =\|\bphi_{\mathrm{B},i}\|^{2}=1$. Thus, the $ m $th AP receives
		\begin{align}
			\bY_{\mathrm{p},m}=\sqrt{\tau_{\mathrm{p}}p_{\mathrm{p}}}\sum_{i=1}^{W}\left(\bh_{mi} \bphi_{\mathrm{A},i}^{\H}+\bg_{mi}\bphi_{\mathrm{B},i}^{\H}\right)+\bW_{\mathrm{p},m},
		\end{align}
		where $ \bW_{\mathrm{p},m} $ is an $ N \times \tau_{\mathrm{p}} $ matrix describing the additive white Gaussian noise (AWGN) and having i.i.d elements distributed as $ \mathcal{CN}\left(0,1\right) $. Also, $ p_{\mathrm{p}} $ expresses the normalized transmit signal-to-noise ratio (SNR) of each pilot symbol.
		
		Following the approach in \cite{Ngo2014}, we obtain the MMSE estimated channels for the $ i $th pairbased on \cite{Verdu1998} as
		\begin{align}
			\bh_{mi}&=\hat{\bh}_{mi}+\tilde{\bh}_{mi}\label{estimated_H}\\
			\bg_{mi}&=\hat{\bg}_{mi}+\tilde{\bg}_{mi},\label{estimated_G}
		\end{align}
		where $ \hat{\bh}_{mi}\sim \mathcal{CN}\left(\b0,\phi_{\mathrm{A},mi}\Id_{N}\right) $ and $ \tilde{\bh}_{mi}\sim \mathcal{CN}\left(\b0,{e}_{\mathrm{A},mi}\Id_{N}\right) $ are the estimated and estimation error channel vectors, being mutually independent with $ \phi_{\mathrm{A},mi} =\frac{\tau_{\mathrm{p}}p_{\mathrm{p}}\al^{2}_{\mathrm{A},mi}}{1+\tau_{\mathrm{p}}p_{\mathrm{p}}\al_{\mathrm{A},mi}}$ and $ e_{\mathrm{A},mi} =\frac{\al_{\mathrm{A},mi}}{1+\tau_{\mathrm{p}}p_{\mathrm{p}}\al_{\mathrm{A},mi}}$, respectively. Similar expressions hold for the estimated channel and estimation error vectors in \eqref{estimated_G}, i.e., $ \hat{\bg}_{mi}\sim \mathcal{CN}\left(\b0,\phi_{\mathrm{B},mi}\Id_{N}\right) $ and $ \tilde{\bg}_{mi}\sim \mathcal{CN}\left(\b0,e_{\mathrm{B},mi}\Id_{N}\right) $ are mutually independent with $ \phi_{\mathrm{B},mi} =\frac{\tau_{\mathrm{p}}p_{\mathrm{p}}\al^{2}_{\mathrm{B},mi}}{1+\tau_{\mathrm{p}}p_{\mathrm{p}}\al_{\mathrm{B},mi}}$ and $e_{\mathrm{B},mi} =\frac{\al_{\mathrm{B},mi}}{1+\tau_{\mathrm{p}}p_{\mathrm{p}}\al_{\mathrm{B},mi}}$, respectively.

		\subsection{Data Transmission}
		The communication takes place in two phases, as described next.
		\subsubsection{Phase I}
		It is known as a MAC phase, where all user pairs simultaneously transmit their data to all relay nodes, being the APs. In other words, the received signal by AP $ m $ from the $ W $ pairs $ \mathrm{T}_{\mathrm{A},i}$, $ \mathrm{T}_{\mathrm{B},i}$ is given by
		\begin{align}
			\!\!\by_{m}\!=\!\sqrt{p_{\mathrm{u}}}\sum_{i=1}^{W}\!\left(\sqrt{\eta_{\mathrm{A},i}}\bh_{mi}q_{\mathrm{A},i}+\sqrt{\eta_{\mathrm{B},i}}\bg_{mi}q_{\mathrm{B},i}\right)+\bn_{\mathrm{u},m},\label{received1}
		\end{align}
		where $ q_{\mathrm{A},i} $ and $ q_{\mathrm{B},i} $ are the data by the $ i $th user pair weighted by the power control coefficients $ {\eta_{\mathrm{A},i}} $ and $ {\eta_{\mathrm{B},i}} $, respectively. Note that $ \EE\{|q_{\mathrm{A},i}|^{2}\}=\EE\{|q_{\mathrm{B},i}|^{2}\}=1 $ while $ 0\le \eta_{\mathrm{A},i}\le 1 $ and $ 0\le \eta_{\mathrm{B},i}\le 1$. Also, $ p_{\mathrm{u}} $ denotes the normalized uplink SNR and $ \bn_{\mathrm{u},m}\sim \mathcal{CN}\left(\b0,\Id_{N}\right) $ is the AWGN vector at AP $ m $.
		
		By accounting for linear detection by means of the linear receiver matrix $ \bW_{m} =\left[\bW_{m,\mathrm{A}}, \bW_{m,\mathrm{B}}\right]\in \mathbb{C}^{N\times 2W}$, the $ m $th AP multiplies its received signal $ \by_{m} $ with the transpose of the linear detector and obtains the post-processed signals as
		\begin{align}
			\br_{m}=\begin{bmatrix}
				\bW_{m,\mathrm{A}}^{\H} \by_{m}\\
				\bW_{m,\mathrm{B}}^{\H} \by_{m}
			\end{bmatrix}.
		\end{align}
		
		Herein, we assume application of MRC due to its low complexity. Also, it is suggested for implementation in a distributed fashion (locally at the APs), and it results in simplified expressions as suggested initially in \cite{Ngo2017}. Note that the top $ W $ rows of $ \br_{m} $ correspond to $ W $ signals from $ \mathrm{T}_{\mathrm{A}}$ while the remaining $ W $ bottom rows stand for the signals from $ \mathrm{T}_{\mathrm{B},i}$ $ \left(i=1,\ldots,W\right) $. Next, all APs sent their processed signals to the CPU through perfect backhaul links and the CPU obtains
	\begin{align}
				\br&= \sum_{m=1}^{M}\br_{m} \\
				&=\begin{bmatrix}
					\sum_{m=1}^{M}	\bW_{m,\mathrm{A}}^{\H} \by_{m}\\
					\sum_{m=1}^{M} \bW_{m,\mathrm{B}}^{\H} \by_{m}
				\end{bmatrix}.
			\end{align}
			Now, we focus on the detection of the transmitted symbols of the $ i $th user pair at the CPU. Specifically, we denote 
			\begin{align}
				r_{m_{i}}&=	 \sum_{m=1}^{M}(\bW_{m,\mathrm{A}}^{\H})_{i} \by_{m}\\
				r_{m_{W+i}}&=\sum_{m=1}^{M} 	(\bW_{m,\mathrm{B}}^{\H})_{i} \by_{m}
			\end{align}
			where $ r_{m_{i}} $ and $ r_{m_{W+i}} $ are the detected symbols from $ \mathrm{T}_{\mathrm{A},i}$ and $ \mathrm{T}_{\mathrm{B},i}$ , respectively. Note that $ (\bX)_{i} $ denotes the $ i $th row of $ \bX $. Thus, the total detected symbols at the CPU from both $ \mathrm{T}_{\mathrm{A},i}$ and $ \mathrm{T}_{\mathrm{A},i}$ are given by 
			\begin{align}
				\tilde{r}_{i}&=\sum_{m=1}^{M}r_{m_{i}}+r_{m_{W+i}}\\
				&=\sum_{m=1}^{M}\left(\hatvh^{\H}_{mi}+\hatvg^{\H}_{mi}\right)\by_{m}.\label{MacReceived}
			\end{align}
		\subsubsection{Phase II}
		This phase includes encoding of the received information and broadcasting it to all user pairs. Specifically, AP $ m $ applies linear precoding matrices in terms of MRT given by $ \hat{\bG}_{m}$ and $ \hat{\bH}_{m} $ to transmit the signals $ \bq_{\mathrm{A}}$ and $ \bq_{\mathrm{B}} $ to $ \mathrm{T}_{\mathrm{A},i} $ and $ \mathrm{T}_{\mathrm{B},i} $ by using the power control coefficients $\etav_{\mathrm{A},m}$ and $ \etav_{\mathrm{B},m} $, respectively. Thus, the transmit signal is written as
		\begin{align}
			\bx_{m}=\sqrt{p_{\mathrm{d}}}\left(\hat{\bG}_{m}^{*}\etav_{\mathrm{A},m}^{1/2}\bq_{\mathrm{A}}+\hat{\bH}_{m}^{*}\etav_{\mathrm{B},m}^{1/2}\bq_{\mathrm{B}}\right),
		\end{align}
		where $p_{\mathrm{d}}$ is the maximum normalized transmit power while for notational convenience, we have denoted $ \hat{\bG}_{m}=\left[\hat{\bg}_{m1}, \ldots, \hat{\bg}_{mW}\right] \in\mathbb{C}^{N \times W}$, $ \hat{\bH}_{m}=\left[\hatvh_{m1}, \ldots, \hatvh_{mW}\right] \in\mathbb{C}^{N \times W}$, $ \etav_{\mathrm{A},m}=\diag\left(\eta_{\mathrm{A},m1}, \ldots, \eta_{\mathrm{A},mW}\right)$, $ \etav_{\mathrm{B},m}=\diag\left(\eta_{\mathrm{B},m1}, \ldots, \eta_{\mathrm{B},mW}\right)$, $ \bq_{\mathrm{A}}=\left[q_{\mathrm{A},1}, \ldots, q_{\mathrm{A},W}\right]^{\T}$, and $ \bq_{\mathrm{B}}=\left[q_{\mathrm{B},1}, \ldots, q_{\mathrm{B},W}\right]^{\T}$. Note that the power control coefficients are chosen to satisfy the power constraint, $ \EE\{\|\bx_{m}\|^{2}\}\le p_{\mathrm{d}}$, which gives
		\begin{align}
			N\sum_{i=1}^{W} \left(\eta_{\mathrm{A},mi} \phi_{\mathrm{B},mi}+\eta_{\mathrm{B},mi} \phi_{\mathrm{A},mi}\right)\le 1.\label{prPower1}
		\end{align}
		
		Also, we have a total power constraint for all APs acting as relay nodes, i.e., $ \EE\{\sum^{M}_{m=1}\|\bx_{m}\|^{2}\}= p_{\mathrm{r}}$, which results in
		\begin{align}
			p_{\mathrm{d}}=\frac{p_{\mathrm{r}}}{N\sum_{i=1}^{W}\sum^{M}_{m=1}\left(\eta_{\mathrm{A},mi} \phi_{\mathrm{B},mi}+\eta_{\mathrm{B},mi} \phi_{\mathrm{A},mi}\right)}.\label{prPower}
		\end{align}

		The received signal at $ \mathrm{T}_{\mathrm{\mathrm{A}},i} $ by all APs is written as
		\begin{align}
			&z_{\mathrm{A},i}=\sum_{m=1}^{M}\bh_{mi}^{\T}\bx_{m}+n_{\mathrm{A},i} \nn\\
			&\!=\!\sqrt{p_{\mathrm{d}}}\!\sum_{m=1}^{M}\!\sum_{j=1}^{W}\nn\bh_{mi}^{\T}\!\left(\!\eta_{\mathrm{A},mj}^{1/2}\hat{\bg}_{mj}^{*}q_{\mathrm{A},j}\!+\!\eta_{\mathrm{B},mj}^{1/2}\hat{\bh}_{mj}^{*}q_{\mathrm{B},j}\!\right)\!+\!n_{\mathrm{A},i}. \end{align}
		
		Similarly, $ \mathrm{T}_{\mathrm{\mathrm{B}},i} $ receives by all APs
		\begin{align}
			z_{\mathrm{B},i}&\!=\!\sqrt{p_{\mathrm{d}}}\!\sum_{m=1}^{M}\!\sum_{i=1}^{W}\!\bg_{mi}^{\T}\!\left(\!\eta_{\mathrm{A},mj}^{1/2}\hat{\bg}_{mj}^{*}q_{\mathrm{A},j}\!+\!\eta_{\mathrm{B},mj}^{1/2}\hat{\bh}_{mj}^{*}q_{\mathrm{B},j}\!\right)\!\nn\\
			&+\!n_{\mathrm{B},i}.
		\end{align}
		Note that $ \bn_{\mathrm{X},i} \sim\mathcal{CN}\left(\b0,\Id_{N}\right)$ is the additive noise at $ \mathrm{T}_{\mathrm{\mathrm{X}},i} $ ($ \mathrm{X}\in\{\mathrm{A},\mathrm{B}\} $).

		\section{SE Analysis}\label{SEA}
		This section presents the SE performance analysis of the DF two-way CF mMIMO with MRC/MRT linear processing by means of exact and closed-form expressions. 
		
		\begin{figure*}
			\begin{align}
				\tilde{r}_{i}&=\sqrt{p_{\mathrm{u}}}\bigg(\underbrace{\sqrt{\eta_{\mathrm{A},i}}\EE\Big\{\sum_{m=1}^{M}\left(\hatvh^{\H}_{mi}\bh_{mi}+\hatvg^{\H}_{mi}\bh_{mi}\right)\Big\}q_{\mathrm{A},i}+\sqrt{\eta_{\mathrm{B},i}}\EE\Big\{\sum_{m=1}^{M}\left(\hatvh^{\H}_{mi}\bg_{mi}+\bg^{\H}_{mi}\hatvg_{mi}\right)\Big\}q_{\mathrm{B},i}}_\text{desired signal}\nn\\
				&+\underbrace{\sqrt{\eta_{\mathrm{A},i}}\sum_{m=1}^{M}\left(\hatvh^{\H}_{mi}\bh_{mi}+\hatvg^{\H}_{mi}\bh_{mi}\right)q_{\mathrm{A},i}-\sqrt{\eta_{\mathrm{A},i}}\EE\Big\{\sum_{m=1}^{M}\left(\hatvh^{\H}_{mi}\bh_{mi}+\hatvg^{\H}_{mi}\bh_{mi}\right)\Big\}q_{\mathrm{A},i}}_\text{estimation error,~A}\nn\\
				&+\sqrt{\eta_{\mathrm{B},i}}\underbrace{\sum_{m=1}^{M}\sum_{m=1}^{M}\left(\hatvh^{\H}_{mi}\bg_{mi}+\hatvg^{\H}_{mi}\bg_{mi}\right)q_{\mathrm{B},i}-\sqrt{\eta_{\mathrm{B},i}}\EE\Big\{\sum_{m=1}^{M}\left(\hatvh^{\H}_{mi}\bg_{mi}+\hatvg^{\H}_{mi}\bg_{mi}\right)\Big\}q_{\mathrm{B},i}}_\text{estimation error,~B}\nn\\
				&+\underbrace{\sum_{j\ne i}^{W}\sum_{m=1}^{M}\sqrt{\eta_{\mathrm{A},j}}\left(\hatvh^{\H}_{mi}\bh_{mj}+\hatvg^{\H}_{mi}\bh_{mj}\right)q_{\mathrm{A},j}+\sum_{j\ne i}^{W}\sum_{m=1}^{M}\sqrt{\eta_{\mathrm{B},j}}\left(\hatvh^{\H}_{mi}\bg_{mj}+\hatvg^{\H}_{mi}\bg_{mj}\right)q_{\mathrm{B},j}}_\text{inter-user interference}\bigg)\nn\\
				&+\underbrace{\sum_{m=1}^{M}\left(\hatvh^{\H}_{mi}+\hatvg^{\H}_{mi}\right)\bn_{\mathrm{u},m}}_\text{post-pocessed noise}.\label{received2}
			\end{align}
			\hrulefill
		\end{figure*}
		\subsubsection{Phase I}
		The received signal, given by \eqref{MacReceived} can be written as in~\eqref{received2} at the top of next page after substituting the received signal given by \eqref{received1}. Also, we have used \eqref{estimated_H} and \eqref{estimated_G} since the $ m $th AP has imperfect CSI and considers the estimated channels as its true channels. Taking as a reference the $ i $th user pair, we obtain its achievable SE in the MAC phase by means of the
		use-and-then-forget capacity bounding technique where the CPU uses only statistical knowledge
		of the channel when performing the detection and the unknown terms are treated as uncorrelated additive noise~\cite[Ch. 3]{Marzetta2016}, \cite{Bjoernson2017}. Note that, in the case of mMIMO, this bound exploits channel hardening and becomes tighter as the number of antennas increases. Relied on this assumption, many works in CF mMIMO exploited that channel hardening appears in the case of a large number of APs. However, in~\cite{Chen2018}, it was shown that, in general, channel hardening is not met in CF mMIMO systems with single-antenna APs, but it appears if multi-antenna APs are considered. Numerical results in Section~\ref{Numerical}, relied on this assumption, verify the tightness of this bound. Thus, the achievable SE is given by \begin{align}
			R_{i}^\text{MAC}=\frac{\tau_{\mathrm{c}}-\tau_{\mathrm{p}}}{2 \tau_{\mathrm{c}}}\log_{2}\left(1+\gamma_{i}^\text{MAC}\right),\label{received6}
		\end{align}
		where the corresponding SINR is given by 
		\begin{align}
			\gamma_{i}^\text{MAC}=\frac{\mathrm{DS}_{A_{i}}^\text{MAC}+\mathrm{DS}_{B_{i}}^\text{MAC}}{\mathrm{EE}_{\mathrm{A},i}^\text{MAC}+\mathrm{EE}_{\mathrm{B},i}^\text{MAC}+\mathrm{IUI}_{i}^\text{MAC}+\mathrm{N}_{i}^\text{MAC}}\label{gammaMac}
		\end{align}
		with the various terms provided by
		\begin{align}
			\mathrm{DS}_{A_{i}}^\text{MAC}&=\eta_{\mathrm{A},i}\Big|\EE\Big\{\sum_{m=1}^{M}\left(\hatvh^{\H}_{mi}\bh_{mi}+\hatvg^{\H}_{mi}\bh_{mi}\right)\!\!\Big\}\Big|^{2},\\
			\mathrm{DS}_{B_{i}}^\text{MAC}&=\eta_{\mathrm{B},i}\Big|\EE\Big\{\sum_{m=1}^{M}\left(\hatvh^{\H}_{mi}\bg_{mi}+\hatvg^{\H}_{mi}\bg_{mi}\right)\!\!\Big\}\Big|^{2},\\
			\mathrm{EE}_{\mathrm{A},i}^\text{MAC}&=\eta_{\mathrm{A},i}\var\Big\{\sum_{m=1}^{M}\left(\hatvh^{\H}_{mi}\bh_{mi}+\hatvg^{\H}_{mi}\bh_{mi}\right)\!\!\Big\},\nn\\
			\mathrm{EE}_{\mathrm{B},i}^\text{MAC}&=\eta_{\mathrm{B},i}\var\Big\{\sum_{m=1}^{M}\sqrt{\eta_{\mathrm{B},i}}\left(\hatvh^{\H}_{mi}\bg_{mi}+\hatvg^{\H}_{mi}\bg_{mi}\right)\!\!\Big\},\nn\\
			\mathrm{IUI}_{i}^\text{MAC}&=\sum_{j\ne i}^{W} \eta_{\mathrm{A},j}
			\EE\Big\{\Big|\sum_{m=1}^{M}\hatvh^{\H}_{mi}\bh_{mj}\Big|^{2}+\Big|\sum_{m=1}^{M}\hatvg^{\H}_{mi}\bh_{mj}\Big|^{2}\Big\}\nn\\
			&\!\!\!+\sum_{j\ne i}^{W} \eta_{\mathrm{B},j}
			\EE\Big\{\Big|\sum_{m=1}^{M}\hatvh^{\H}_{mi}\bg_{mj}\Big|^{2}+\Big|\sum_{m=1}^{M}\hatvg^{\H}_{mi}\bg_{mj}\Big|^{2}\Big\},\\
			\mathrm{N}_{i}^\text{MAC}&=\frac{1}{p_{\mathrm{u}}}\EE\Big\{\|\sum_{m=1}^{M}\hatvh_{mi}\|^{2}+\|\sum_{m=1}^{M}\hatvg_{mi}\|^{2}\Big\}.
		\end{align}
	Note that $ \mathrm{DS}$, $ \mathrm{EE} $, $ \mathrm{IUI}_{ik} $, and $ \mathrm{N} $ express the desired signal (DS) part, the estimation error (EE) part, the inter-user interference (IUI), and the thermal noise and can refer to both MAC and BC phases as well as users $ \mathrm{T}_{\mathrm{A},i}$ and $ \mathrm{T}_{\mathrm{B},i}$.
		Moreover, the achievable SE of the link $ T_{\mathrm{X},i} $ where $ \mathrm{X}\in\{\mathrm{A},\mathrm{B}\} $, is written as
		\begin{align}
			R_{\mathrm{X},i}^\text{MAC}=\frac{\tau_{\mathrm{c}}-\tau_{\mathrm{p}}}{2 \tau_{\mathrm{c}}}\log_{2}\left(1+\gamma_{\mathrm{X},i}^\text{MAC}\right)\label{received8}
		\end{align}
		with signal-to-interference-plus-noise ratio (SINR) given by
		\begin{align}
			\gamma_{\mathrm{X},i}^\text{MAC}=\frac{\mathrm{DS}_{\mathrm{X}_{i}}^\text{MAC}}{\mathrm{EE}_{\mathrm{A},i}^\text{MAC}+\mathrm{IUI}_{i}^\text{MAC}+\mathrm{N}_{i}^\text{MAC}}.
		\end{align}

		\subsubsection{Phase II}
		Since in practice, the users are not aware of the instantaneous CSI, we  take advantage again of the channel hardening and the use-and-then-forget bound from the massive MIMO (mMIMO) literature \cite{Bjoernson2017}. Thus, we assume that user $ \mathrm{T}_{\mathrm{A},i} $ has knowledge only of its statistics
		and performs partial self-interference cancellation to obtain
		
		\begin{align}
			& \hat{z}_{\mathrm{A},i}=z_{\mathrm{A},i}-\sqrt{p_{\mathrm{d}}}\EE\Big\{\sum_{m=1}^{M}\sqrt{ \eta_{\mathrm{A},mi}}\bh_{mi}^{\T}\hat{\bg}_{mi}^{*}\Big\}q_{\mathrm{A},i}\\
			&=\sqrt{p_{\mathrm{d}}}\EE\Big\{\sum_{m=1}^{M}\sqrt{ \eta_{\mathrm{B},mi}}\bh_{mi}^{\T}\hat{\bh}_{mi}^{*}\Big\}q_{\mathrm{B},i}\nn\\
			&\!+\!\sqrt{p_{\mathrm{d}}}\!\left(\sum_{m=1}^{M}\!\!\!\sqrt{ \eta_{\mathrm{B},mi}}\bh_{mi}^{\T}\hat{\bh}_{mi}^{*}\!-\!\EE\Bigg\{\!\sum_{m=1}^{M}\!\!\sqrt{ \eta_{\mathrm{B},mi}}\bh_{mi}^{\T}\hat{\bh}_{mi}^{*}\!\!\Bigg\}\!\!\right)
			\!\!q_{\mathrm{B},i} \nn \\
			&\!+\!\sqrt{p_{\mathrm{d}}}\!\left(\sum_{m=1}^{M}\!\!\!\sqrt{ \eta_{\mathrm{A},mi}}\bh_{mi}^{\T}\hat{\bg}_{mi}^{*}\!-\!\EE\Bigg\{\!\sum_{m=1}^{M}\!\!\sqrt{ \eta_{\mathrm{A},mi}}\bh_{mi}^{\T}\hat{\bg}_{mi}^{*}\!\!\Bigg\}\!\!\right)
			\!\!q_{\mathrm{A},i}\nn\\
			&+\sqrt{p_{\mathrm{d}}}\sum_{j\ne i}^{W}\!\!\sum_{m=1}^{M}\left(\!\sqrt{\eta_{\mathrm{A},mj}}\bh_{mi}^{\T}\hat{\bg}_{mj}^{*}
			q_{\mathrm{A},j}+\sqrt{\eta_{\mathrm{B},mj}}\bh_{mi}^{\T}\hat{\bh}_{mj}^{*}
			q_{\mathrm{B},j}\!\right)\nn\\
			&+n_{\mathrm{A},i},\label{received3}
		\end{align}
		where the first term expresses the desired signal, the second and third terms express the gain uncertainty, the fourth and fifth terms represent the residual self-interference, the sixth and seventh terms describe the inter-pair interference while the last term denotes the noise.
		The achievable SE of $ \mathrm{T}_{\mathrm{A},i} $ during the BC phase is obtained by
		\begin{align}
			R_{\mathrm{A},i}^\text{BC}=\frac{\tau_{\mathrm{c}}-\tau_{\mathrm{p}}}{2 \tau_{\mathrm{c}}}\log_{2}\left(1+\gamma_{\mathrm{A},i}^\text{BC}\right)\!,
		\end{align}
		where $ \gamma_{\mathrm{A},i}^\text{BC} $ is given by
		\begin{align}
			\!\!\!\gamma_{\mathrm{A},i}^\text{BC}=\frac{\mathrm{DS}_{A_{i}}^\text{BC}}{ \mathrm{BU}_{A_{i}}^\text{BC}+ \mathrm{BU}_{B_{i}}^\text{BC}+\sum_{j\ne i}\!\left( \mathrm{IUI}_{A_{j}}^\text{BC}+ \mathrm{IUI}_{B_{j}}^\text{BC}\right)+\frac{1}{p_{\mathrm{d}}}}\label{received5}
		\end{align}
		with
		\begin{align}
			\mathrm{DS}_{A_{i}}^\text{BC}&=\Big|\EE\Big\{\sum_{m=1}^{M}\sqrt{ \eta_{\mathrm{B},mi}}\bh_{mi}^{\T}\hat{\bh}_{mi}^{*}\Big\}\Big|^{2},\label{BC1}\\
			\mathrm{BU}_{A_{i}}^\text{BC}&=\var\Big\{\sum_{m=1}^{M}\sqrt{ \eta_{\mathrm{B},mi}}\bh_{mi}^{\T}\hat{\bh}_{mi}^{*}\Big\},\label{BC2}\\
			\mathrm{BU}_{B_{i}}^\text{BC}&=\var\Big\{\sum_{m=1}^{M}\sqrt{ \eta_{\mathrm{A},mi}}\bh_{mi}^{\T}\hat{\bg}_{mi}^{*}\Big\},\label{BC3}\\
			\mathrm{IUI}_{A_{i}}^\text{BC}&=\EE\Big\{\Big|\sum_{m=1}^{M}\sqrt{ \eta_{\mathrm{A},mj}}\bh_{mi}^{\T}\hat{\bg}_{mj}^{*}\Big|^{2}\Big\},\label{BC4}\\
			\mathrm{IUI}_{B_{i}}^\text{BC}&=\EE\Big\{\Big|\sum_{m=1}^{M}\sqrt{ \eta_{\mathrm{B},mj}}\bh_{mi}^{\T}\hat{\bh}_{mj}^{*}\Big|^{2}\Big\}.\label{BC5}
		\end{align}
		Note that $ \mathrm{BU} $ refers to the beamforming gain uncertainty (BU). Similarly, we obtain the achievable SE of $ \mathrm{T}_{\mathrm{B},i}$, $ R_{\mathrm{B},i}^\text{BC} $, after obtaining the post-processed signal at user $ \mathrm{T}_{\mathrm{B},i} $, $ \hat{z}_{\mathrm{B},i} $, by means of a similar expression to \eqref{received3}. Hence, the achievable SE of the $ i $th pair $ \mathrm{T}_{\mathrm{A},i}$ to $ \mathrm{T}_{\mathrm{B},i}$ is given by $ \min\left(R_{\mathrm{A},i}^\text{MAC}, R_{\mathrm{B},i}^\text{BC}\right) $ while the SE for the opposite direction is $ \min\left(R_{\mathrm{B},i}^\text{MAC}, R_{\mathrm{A},i}^\text{BC}\right) $ with the individual SEs obtained previously. As a result, the achievable SE of of the $ i $th pair during BC is given by the sum
		\begin{align}
			R_{i}^\text{BC}= \min\left(R_{\mathrm{A},i}^\text{MAC}, R_{\mathrm{B},i}^\text{BC}\right) +\min\left(R_{\mathrm{B},i}^\text{MAC}, R_{\mathrm{A},i}^\text{BC}\right).\label{received7}
		\end{align}
		
		According to \cite{Rankov2007,Gao2013}, the achievable SE of the $ i $th pair over both phases is given by \eqref{received6} and \eqref{received7} as
		\begin{align}
			R_{i}=\min\left(R_{i}^\text{MAC},R_{i}^\text{BC}\right)\!,
		\end{align}
		and the achievable sum SE of a multipair two-way CF mMIMO relaying system is given by
		\begin{align}
			R=\sum_{i=1}^{W}R_{i}.\label{TotalSE}
		\end{align}

		\begin{theorem}\label{theoremTotalSE} 
			The achievable sum SE of a multipair two-way CF mMIMO relaying system with DF protocol and MRC/MRT linear processing, for any finite $ 	M $ and $ W $, is given by \eqref{TotalSE} including the SEs provided by \eqref{MAC11}-\eqref{MAC1111} at the top of the next page with $ \bar{\mathrm{X}} $ being the complement of $ \mathrm{X}\in\{\mathrm{A},\mathrm{B}\} $.
			\begin{figure*}
				\begin{align}
					R_{i}^\text{MAC}&=\frac{\tau_{\mathrm{c}}-\tau_{\mathrm{p}}}{2 \tau_{\mathrm{c}}}\log_{2}\left(1+\frac{p_{\mathrm{u}}N\left(\left(\sum_{m=1}^{M}\sqrt{\eta_{\mathrm{A},i}}\phi_{\mathrm{A},mi}\right)^{2}+\left(\sum_{m=1}^{M}\sqrt{\eta_{\mathrm{B},i}}\phi_{\mathrm{B},mi}\right)^{2}\right)}{\sum_{m=1}^{M}\left(\sum_{j=1 }^{W}p_{\mathrm{u}}\left( \eta_{\mathrm{A},j}\al_{\mathrm{A},mj}+ \eta_{\mathrm{B},j}\al_{\mathrm{B},mj}\right)+1\right) \left(\phi_{\mathrm{A},mi}+\phi_{\mathrm{B},mi}\right)}\right)\!,\label{MAC11}\\
					R_{\mathrm{X},i}^\text{MAC}&=\frac{\tau_{\mathrm{c}}-\tau_{\mathrm{p}}}{2 \tau_{\mathrm{c}}}\log_{2}\left(1+\frac{p_{\mathrm{u}}\eta_{\mathrm{A},i}N\left(\sum_{m=1}^{M}\phi_{\mathrm{X},mi}\right)^{2}}{\sum_{m=1}^{M}\left(\sum_{j=1 }^{W}p_{\mathrm{u}}\left( \eta_{\mathrm{A},j}\al_{\mathrm{A},mj}+ \eta_{\mathrm{B},j}\al_{\mathrm{B},mj}\right)+1\right) \left(\phi_{\mathrm{A},mi}+\phi_{\mathrm{B},mi}\right)}\right)\!,\\
					R_{\mathrm{X},i}^\text{BC}&=\frac{\tau_{\mathrm{c}}-\tau_{\mathrm{p}}}{2 \tau_{\mathrm{c}}}\log_{2}\left(1+\frac{Np_{\mathrm{r}}\left(\sum_{m=1}^{M}\sqrt{\eta_{\mathrm{\bar{X}},mi}}\phi_{\mathrm{X},mi}\right)^{2}}{\sum_{m=1}^{M}\sum_{j=1}^{W}\left( p_{\mathrm{r}}\al_{\mathrm{X},mi}+1\right) \left(\eta_{\mathrm{A},mj}\phi_{\mathrm{B},mj}+\eta_{\mathrm{B},mj}\phi_{\mathrm{A},mj}\right)}\right)\!.\label{MAC1111}
				\end{align}
				\hrulefill
			\end{figure*}
		\end{theorem}
		\begin{proof}
			See Appendix~\ref{TotalSEproof}.
		\end{proof}
		
		Notably, our procedure results in exact closed-form results while other similar works rely on approximations. Also, regarding the dependence of the individual SEs with respect to the transmit power, we observe that they are interference-limited as expected \cite{Kong2018}. Furthermore, we notice an increase of the SEs with the number of APs by taking advantage of the CF mMIMO architecture, which is studied in depth below.
		
		\section{Power Efficiency}\label{PE}
		Herein, we present a detailed study concerning the achievable power savings by letting the number of APs grow large, i.e., $ M\to \infty $. These savings are known as power-scaling laws that allow preserving a specific SE while reducing the transmit powers. 
		
		We assume that all the users have the same transmit power, i.e., no power control is considered. In particular, let $ \eta_{\mathrm{A},i}=\eta_{\mathrm{B},i}=1$ for the sake of simplicity, we are going to shed light on the impact of the transmit power per user, pilot symbol, and relay on the separate SEs in the large APs limit, i.e., $ M \to \infty $.
		\subsection{Scenario A: $ p_{\mathrm{p}}= \frac{E_{\mathrm{p}}}{M^{\al} }$, and fixed $ p_{\mathrm{r}}, p_{\mathrm{u}} $}
		This scenario concerns the study of the power efficiency in the training phase. 
		\begin{proposition}\label{PropositionPEproof}
			For fixed $ p_{\mathrm{r}}, p_{\mathrm{u}} $, and $ E_{\mathrm{p}} $, when $ p_{\mathrm{p}}= \frac{E_{\mathrm{p}}}{M^{\al} }$ with $ \al>0 $ and $ M \to \infty $, we obtain
			\begin{align}
				\gamma_{i}^\text{MAC}&\!=\!\frac{p_{\mathrm{u}}N\frac{E_{\mathrm{p}}}{M^{\al} }\!\!\left(\!\!\left(\!\sum_{m=1}^{M}\!\al_{\mathrm{A},mi}\right)^{\!2}\!+\!\left(\sum_{m=1}^{M}\!\al_{\mathrm{B},mi}\right)^{\!2}\right)}{\!\!\!\displaystyle\sum_{m=1}^{M}\!\!\!\left( \sum_{j=1}^{W}\!\!\left(\!\al_{\mathrm{A},mj}\!+\! \al_{\mathrm{B},mj}\!\right)\!+\!1\!\!\right)\!\! \left(\!\al_{\mathrm{A},mi}\!+\!\al_{\mathrm{B},mi}\!\right)}\label{gammaMac1},\!\!\!\\
				\gamma_{X,i}^\text{MAC}&= \frac{p_{\mathrm{u}}N\frac{E_{\mathrm{p}}}{M^{\al} }\left(\sum_{m=1}^{M}\al_{\mathrm{X},mi}\right)^{2}}{\!\!\!\displaystyle\sum_{m=1}^{M}\!\!\!\left( \sum_{j=1}^{W}\!\!\left(\!\al_{\mathrm{A},mj}\!+\! \al_{\mathrm{B},mj}\!\right)\!+\!1\!\!\right)\!\! \left(\!\al_{\mathrm{A},mi}\!+\!\al_{\mathrm{B},mi}\!\right)},\label{MAC1121}\\
				\gamma_{\mathrm{X},i}^\text{BC}&\!=\!\frac{Np_{\mathrm{r}}\frac{E_{\mathrm{p}}}{M^{\al} }\left(\sum_{m=1}^{M}\al_{\mathrm{X},mi}\right)}{\displaystyle\sum_{m=1}^{M}\sum_{j=1}^{W}\!\left(p_{\mathrm{r}} \al_{\mathrm{X},mi}+1\right)\! \left(\al_{\mathrm{B},mj}+\al_{\mathrm{A},mj}\right)}.\label{MAC1131}
			\end{align}
		\end{proposition}
		
		Clearly, the choice of $ \al $ defines the result. For example, in \eqref{gammaMac1}, the order of $ \gamma_{i}^\text{MAC}$ is $\mathcal{O}\left(M^{1-\al}\right) $ , which implies that $ \gamma_{i}^\text{MAC}\to 0 $, if $ \al>1 $. On the other hand, if $ 0<\al<1 $, $ \gamma_{i}^\text{MAC} \to \infty $. However, when $ \al=1 $, the corresponding SINRs result in a finite limit. Similar comments holds for the SINRs given by \eqref{MAC1121} and \ref{MAC1131}. Note that the corresponding non-zero limits are written as
		\begin{align}
			\gamma_{i}^\text{MAC}&=\frac{p_{\mathrm{u}}N{E_{\mathrm{p}}}{ }\!\!\left(\!\!\left(\!\sum_{m=1}^{M}\!\al_{\mathrm{A},mi}\right)^{\!2}\!+\!\left(\sum_{m=1}^{M}\!\al_{\mathrm{B},mi}\right)^{\!2}\right)}{\!\!\!\displaystyle\sum_{m=1}^{M}\!\!\!\left( \sum_{j=1}^{W}\!\!\left(\!\al_{\mathrm{A},mj}\!+\! \al_{\mathrm{B},mj}\!\right)\!+\!1\!\!\right)\!\! \left(\!\al_{\mathrm{A},mi}\!+\!\al_{\mathrm{B},mi}\!\right)},\label{MAC111}\\
			\gamma_{X,i}^\text{MAC}&=\frac{p_{\mathrm{u}}NE_{\mathrm{p}}\left(\sum_{m=1}^{M}\al_{\mathrm{X},mi}\right)^{2}}{\!\!\!\displaystyle\sum_{m=1}^{M}\!\!\!\left( \sum_{j=1}^{W}\!\!\left(\!\al_{\mathrm{A},mj}\!+\! \al_{\mathrm{B},mj}\!\right)\!+\!1\!\!\right)\!\! \left(\!\al_{\mathrm{A},mi}\!+\!\al_{\mathrm{B},mi}\!\right)},\label{MAC112}\\
			\gamma_{\mathrm{X},i}^\text{BC}&\!=\!\frac{Np_{\mathrm{r}}E_{\mathrm{p}}\left(\sum_{m=1}^{M}\al_{\mathrm{X},mi}\right)}{\sum_{m=1}^{M}\sum_{j=1}^{W}\left(p_{\mathrm{r}} \al_{\mathrm{X},mi}+1\right)\! \left(\al_{\mathrm{B},mj}+\al_{\mathrm{A},mj}\right)}\!.\label{MAC113}
		\end{align}
		\subsection{Scenario B: $ p_{\mathrm{u}}= \frac{E_{\mathrm{u}}}{M^{\beta} }$ and $ p_{\mathrm{r}}= \frac{E_{\mathrm{r}}}{M^{\gamma} }$, and fixed $ p_{\mathrm{p}} $}
		This strategy focuses on the power efficiency of solely the data transmission phase. 
		\begin{proposition}\label{PropositionPhaseB}
			For fixed $ p_{\mathrm{p}} $, $ E_{\mathrm{u}} $, and $ E_{\mathrm{r}} $ when $ p_{\mathrm{u}}= \frac{E_{\mathrm{u}}}{M^{\beta} }$, $ p_{\mathrm{r}}= \frac{E_{\mathrm{r}}}{M^{\gamma} }$ with $ \beta\ge 0 $, $ \gamma \ge 0 $, and $ M \to \infty $, we obtain
			\begin{align}
				\gamma_{i}^\text{MAC}&=	\frac{N\frac{ E_{\mathrm{u}}}{M^{\beta}}\left(\displaystyle\left(\sum_{m=1}^{M}\phi_{\mathrm{A},mi}\right)^{2}+\left(\sum_{m=1}^{M}\phi_{\mathrm{B},mi}\right)^{2}\right)}{\sum_{m=1}^{M} \left(\phi_{\mathrm{A},mi}+\phi_{\mathrm{B},mi}\right)},\label{MACB11}\\
				\gamma_{\mathrm{X},i}^\text{MAC}&=\frac{N\frac{E_{\mathrm{u}}}{M^{\beta}}\left(\sum_{m=1}^{M}\phi_{\mathrm{X},mi}\right)^{2}}{\sum_{m=1}^{M} \left(\phi_{\mathrm{A},mi}+\phi_{\mathrm{B},mi}\right)},\label{MACB21}\\
				\gamma_{\mathrm{X},i}^\text{BC}&=\frac{N\frac{E_{\mathrm{r}}}{M^{\gamma}}\left(\sum_{m=1}^{M}\phi_{\mathrm{X},mi}\right)^{2}}{\sum_{m=1}^{M}\sum_{j=1}^{W} \left(\phi_{\mathrm{A},mj}+\phi_{\mathrm{B},mj}\right)}\label{MACB31}.
			\end{align}
		\end{proposition}
		
		It is straightforward to show that the order of both $ \gamma_{i}^\text{MAC} $ and $ \gamma_{\mathrm{X},i}^\text{MAC} $ is $ \mathcal{O}\left(M^{1-\beta}\right) $ while the order of $ \gamma_{\mathrm{X},i}^\text{BC} $ is $ \mathcal{O}\left(M^{1-\gamma}\right) $. As before, the selection of the parameters $ \beta $ and $ \gamma $ affects directly the corresponding SINRs. Especially, we observe that if $ 0\le \beta<1 $ 	and $ 0\le \gamma<1 $, the SINRs grow unboundedly. Also, we notice that under certain circumstances, being equivalent to reducing further the transmit powers of each user or/and the relay by means of $ \beta>1 $ or/and $ \gamma>1 $, the sum SE of the $ i $th pair $ R_{i} $ tends to zero because $ R_{i}^\text{MAC} $ or/and $ R_{i}^\text{BC} $ tend to zero. Especially, $ R_{i} \to 0$ when one of the conditions 1) $ \beta \ge 0 $ and $ \gamma > 1 $, 2) $ \beta >1 $ and $ \gamma \ge 0 $, 3) $ \beta > 1 $ 	and $ \gamma > 1 $ is met. As a result, in order to make $ 	\gamma_{i}^\text{MAC} $ and $ 	\gamma_{X,i}^\text{MAC} $ converge to a non-zero limit, we should have $ \beta=1 $ while when $ \gamma=1 $ ,		$ \gamma_{X,i}^\text{BC} $ takes a finite value. These limits are given by
		\begin{align}
			\gamma_{i}^\text{MAC}&\!=\!	\frac{NE_{\mathrm{u}}\!\!\left(\!\!\left(\sum_{m=1}^{M}\phi_{\mathrm{A},mi}\right)^{\!2}\!+\!\left(\sum_{m=1}^{M}\phi_{\mathrm{B},mi}\right)^{\!2}\right)}{\sum_{m=1}^{M} \left(\phi_{\mathrm{A},mi}+\phi_{\mathrm{B},mi}\right)},\label{MACB1}\\
			\gamma_{\mathrm{X},i}^\text{MAC}&=\frac{NE_{\mathrm{u}}\left(\sum_{m=1}^{M}\phi_{\mathrm{X},mi}\right)^{2}}{\sum_{m=1}^{M} \left(\phi_{\mathrm{A},mi}+\phi_{\mathrm{B},mi}\right)},\label{MACB2}\\
			\gamma_{\mathrm{X},i}^\text{BC}&=\frac{NE_{\mathrm{r}}\left(\sum_{m=1}^{M}\phi_{\mathrm{X},mi}\right)^{2}}{\sum_{m=1}^{M}\sum_{j=1}^{W} \left(\phi_{\mathrm{A},mj}+\phi_{\mathrm{B},mj}\right)}\label{MACB3}.
		\end{align}

		It is worthwhile to mention that Proposition \ref{PropositionPhaseB} reveals that in the large number of APs limit, the reduction of both the transmit power of the APs and the users proportionally to $ M^{-1} $ cancels out the effects of inter-user interference, residual interference, and estimation error. The following corollaries show how the sum SE $ R_{i} $ changes by varying $ \beta $ and $ \gamma$.
		\begin{corollary}
			When $ \beta=1 $ and $ 0\le \gamma<1 $, the SE of the $ i $th user pair as $ M \to \infty $ is written as
			\begin{align}
				&R_{i}=\frac{\tau_{\mathrm{c}}-\tau_{\mathrm{p}}}{2 \tau_{\mathrm{c}}}\nn\\
				&\!\times\!\log_{2}\!\left(\!\!1\!+\!\frac{\displaystyle NE_{\mathrm{u}}\!\!\left(\!\!\left(\sum_{m=1}^{M}\!\!\phi_{\mathrm{A},mi}\right)^{\!\!2}+\left(\sum_{m=1}^{M}\!\!\phi_{\mathrm{B},mi}\right)^{\!\!2}\right)}{\sum_{m=1}^{M} \left(\phi_{\mathrm{A},mi}+\phi_{\mathrm{B},mi}\right)}\right)\!.
			\end{align}
		\end{corollary} 
		
		According to this corollary, $ R_{i} $ is equal to $ R_{i}^\text{MAC} $ since $ R_{i}^\text{BC}=0 $, which means that the SE of the $ i $th user pair depends only on Phase I (MAC phase). The explanation relies on the fact that since we have reduced the transmit power per user much less compared to the transmit power of the APs acting as relays, the MAC phase will present lower performance. Remarkably, in this case, $ R_{i} $ does not depend on the number of users. Also, this SE does not depend on $ E_{\mathrm{r}} $, but it increases with $ E_{\mathrm{u}} $.
		\begin{corollary}\label{corollary2}
			When $ 0\le \beta<1 $ and $ \gamma=1 $, the SE of the $i$th user pair as $ M \to \infty $ is written as
			\begin{align}
				&R_{i}=\frac{\tau_{\mathrm{c}}-\tau_{\mathrm{p}}}{2 \tau_{\mathrm{c}}}\log_{2}\left(1+\frac{NE_{\mathrm{r}}\left(\sum_{m=1}^{M}\phi_{\mathrm{A},mi}\right)^{2}}{\displaystyle \sum_{m=1}^{M}\sum_{j=1}^{W} \left(\phi_{\mathrm{A},mj}+\phi_{\mathrm{B},mj}\right)}\right)\nn\\
				&+\frac{\tau_{\mathrm{c}}-\tau_{\mathrm{p}}}{2 \tau_{\mathrm{c}}}\log_{2}\left(1+\frac{NE_{\mathrm{r}}\left(\sum_{m=1}^{M}\phi_{\mathrm{B},mi}\right)^{2}}{\displaystyle \sum_{m=1}^{M}\sum_{j=1}^{W} \left(\phi_{\mathrm{A},mj}+\phi_{\mathrm{B},mj}\right)}\right)\!.
			\end{align}
		\end{corollary} 
		
		Corollary \ref{corollary2} denotes that the SE of the $i$th user pair appears a bottleneck in the BC phase because the transmit power of the APs during this phase has been cut down more than the transmit power of each user. Herein, we notice that $ R_{i} $ decreases with the number of user pairs $ W$ while it increases with $ E_{\mathrm{r}} $ and it is independent of $ E_{\mathrm{u}} $.
		
		\begin{corollary}\label{corollary3}
			When $ \beta=\gamma=1 $, the SE of the $i$th user pair as $ M \to \infty $ is written as
			\begin{align}
				R_{i}=\min\left(R_{i}^\text{MAC},R_{i}^\text{BC}\right),
			\end{align}
			where 
			\begin{align}
				&R_{i}^\text{MAC}=\frac{\tau_{\mathrm{c}}-\tau_{\mathrm{p}}}{2 \tau_{\mathrm{c}}}\nn\\
				&	\!\times \!\log_{2}\!\left(\!\!1\!+\!\frac{\displaystyle NE_{\mathrm{u}}\!\!\left(\!\!\left(\sum_{m=1}^{M}\!\!\phi_{\mathrm{A},mi}\!\right)^{\!\!2}\!+\!\left(\sum_{m=1}^{M}\!\!\phi_{\mathrm{B},mi}\!\right)^{\!\!2}\right)}{\sum_{m=1}^{M} \left(\phi_{\mathrm{A},mi}+\phi_{\mathrm{B},mi}\right)}\right)\\
				&	R_{i}^\text{BC}= \min\left(R_{\mathrm{A},i}^\text{MAC}, R_{\mathrm{B},i}^\text{BC}\right) +\min\left(R_{\mathrm{B},i}^\text{MAC}, R_{\mathrm{A},i}^\text{BC}\right),
			\end{align}
			with $ R_{\mathrm{X},i}^\text{MAC} $ and $ R_{\mathrm{X},i}^\text{BC} $ given by using \eqref{MACB2} and \eqref{MACB3}, respectively.
		\end{corollary} 
		
		This corollary indicates that if we reduce the transmit power of the APs and users simultaneously and equally to $ 1/M $, both the MAC and BC affect $ R_{i} $.
		
		\subsection{Scenario C: $ p_{\mathrm{p}}= \frac{E_{\mathrm{p}}}{M^{\al} }$, $ p_{\mathrm{u}}= \frac{E_{\mathrm{u}}}{M^{\beta} }$, and $ p_{\mathrm{r}}= \frac{E_{\mathrm{r}}}{M^{\gamma} }$}
		Such a scenario is the most general, where we can achieve power savings in both training and data transmission phases.
		\begin{proposition}\label{PropositionPhaseC}
			When $ p_{\mathrm{p}}= \frac{E_{\mathrm{p}}}{M^{\al} }$ $ p_{\mathrm{u}}= \frac{E_{\mathrm{u}}}{M^{\beta} }$, and $ p_{\mathrm{r}}= \frac{E_{\mathrm{r}}}{M^{\gamma} }$ with $ \al\ge 0 $, $ \beta \ge 0 $, $ \gamma \ge 0 $, and $ E_{\mathrm{p}} $, $ E_{\mathrm{u}} $, $ E_{\mathrm{r}} $ constants, as $ M \to \infty $, we obtain
			\begin{align}
				\gamma_{i}^\text{MAC}&\!=\!	\frac{\displaystyle N\frac{E_{\mathrm{p}}E_{\mathrm{u}}}{M^{\al+\beta}}\!\!\left(\left(\sum_{m=1}^{M}\!\!\al_{\mathrm{A},mi}\right)^{\!\!2}+\left(\sum_{m=1}^{M}\!\!\al_{\mathrm{B},mi}\right)^{\!\!2}\right)}{\sum_{m=1}^{M} \left(\al_{\mathrm{A},mi}+\al_{\mathrm{B},mi}\right)},\label{MACB311}\\
				\gamma_{\mathrm{X},i}^\text{MAC}&=\frac{N\frac{E_{\mathrm{p}}E_{\mathrm{u}}}{M^{\al+\beta}}\left(\sum_{m=1}^{M}\al_{\mathrm{X},mi}\right)^{2}}{\sum_{m=1}^{M} \left(\al_{\mathrm{A},mi}+\al_{\mathrm{B},mi}\right)},\label{MACB321}\\
				\gamma_{\mathrm{X},i}^\text{BC}&=\frac{N\frac{E_{\mathrm{p}}E_{\mathrm{r}}}{M^{\al+\gamma}}\left(\sum_{m=1}^{M}\al_{\mathrm{X},mi}\right)^{2}}{\sum_{m=1}^{M}\sum_{j=1}^{W} \left(\al_{\mathrm{A},mj}+\al_{\mathrm{B},mj}\right)}\label{MACB331}.
			\end{align}
		\end{proposition}
		
		Following the same procedure as before, we observe that the orders of the SINRs in the MAC and BC phases are $\mathcal{O}\left(M^{1-\al-\beta}\right) $ and $ \mathcal{O}\left(M^{1-\al-\gamma}\right) $, respectively. Hence, the corresponding SINRs converge to non-zero limits only when $ \al+\beta=1 $ and $ \al+\gamma=1 $. Otherwise, they can grow unboundedly or tend to zero. These limits are given by
		\begin{align}
			\gamma_{i}^\text{MAC}&\!=\!	\frac{\displaystyle NE_{\mathrm{p}}E_{\mathrm{u}}\!\!\left(\!\!\left(\sum_{m=1}^{M}\!\!\al_{\mathrm{A},mi}\right)^{\!\!2}+\left(\sum_{m=1}^{M}\!\!\al_{\mathrm{B},mi}\right)^{\!\!2}\right)}{\sum_{m=1}^{M} \left(\al_{\mathrm{A},mi}+\al_{\mathrm{B},mi}\right)},\label{MACB3111}\\
			\gamma_{\mathrm{X},i}^\text{MAC}&=\frac{NE_{\mathrm{p}}E_{\mathrm{u}}\left(\sum_{m=1}^{M}\al_{\mathrm{X},mi}\right)^{2}}{\sum_{m=1}^{M} \left(\al_{\mathrm{A},mi}+\al_{\mathrm{B},mi}\right)},\label{MACB32}\\
			\gamma_{\mathrm{X},i}^\text{BC}&=\frac{N E_{\mathrm{p}}E_{\mathrm{r}}\left(\sum_{m=1}^{M}\al_{\mathrm{X},mi}\right)^{2}}{\sum_{m=1}^{M}\sum_{j=1}^{W} \left(\al_{\mathrm{A},mj}+\al_{\mathrm{B},mj}\right)}\label{MACB33}.
		\end{align}
		
		The following corollaries present trade-offs between the transmit powers of the pilot symbols and the APs and/or users.
		\begin{corollary}
			When $ \al+\beta=1 $ and $ \beta> \gamma\ge 0 $, the SE of the $i$th user pair as $ M \to \infty $ is written as
			\begin{align}
				&R_{i}=\frac{\tau_{\mathrm{c}}-\tau_{\mathrm{p}}}{2 \tau_{\mathrm{c}}}\nn\\
				&\!\!\!\!\times \!\log_{2}\!\!\left(\!\!1\!+\!\frac{\displaystyle NE_{\mathrm{p}} E_{\mathrm{u}}\!\!\left(\!\!\!\left(\sum_{m=1}^{M}\!\!\al_{\mathrm{A},mi}\!\right)^{\!\!2}\!\!+\!\left(\sum_{m=1}^{M}\!\!\al_{\mathrm{B},mi}\!\right)^{\!\!2}\right)}{\sum_{m=1}^{M} \left(\al_{\mathrm{A},mi}+\al_{\mathrm{B},mi}\right)}\!\!\right)\!.\!
			\end{align}
		\end{corollary} 
		
		The inequality implies that $ \al+\gamma<1 $. Hence, $ \gamma_{\mathrm{X},i}^\text{BC} \to 0 $, and the SE of the $i$th user pair is determined only by the MAC Phase. Clearly, $ R_{i} $ does not depend on the number of user pairs, which results in the increase with $ W $ of the sum SE
		given by \eqref{TotalSE}.
		
		\begin{corollary}
			When $ \al+\gamma=1 $ and $ \gamma>\beta \ge 0 $, the SE of the $i$th user pair as $ M \to \infty $ is written as
			\begin{align}
				\!\!\!\!R_{i}&=\frac{\tau_{\mathrm{c}}-\tau_{\mathrm{p}}}{2 \tau_{\mathrm{c}}}\log_{2}\left(1+\frac{NE_{\mathrm{p}}E_{\mathrm{r}}\left(\sum_{m=1}^{M}\al_{\mathrm{A},mi}\right)^{2}}{\displaystyle \sum_{m=1}^{M}\sum_{j=1}^{W} \left(\al_{\mathrm{A},mj}+\al_{\mathrm{B},mj}\right)}\right)\nn\\
				&\!\!\!\!\!+\!\frac{\tau_{\mathrm{c}}-\tau_{\mathrm{p}}}{2 \tau_{\mathrm{c}}}\log_{2}\left(1+\frac{NE_{\mathrm{p}}E_{\mathrm{r}}\left(\sum_{m=1}^{M}\al_{\mathrm{B},mi}\right)^{2}}{\displaystyle\sum_{m=1}^{M}\sum_{j=1}^{W} \left(\al_{\mathrm{A},mj}+\al_{\mathrm{B},mj}\right)}\right)\!.
			\end{align}
		\end{corollary} 
		
		Herein, the inequality suggests that $ \al+\beta<1 $, which means that $ \gamma_{\mathrm{X},i}^\text{MAC} \to 0 $, and only the BC phase defines the SE of the $i$th user pair.
		
		\begin{corollary}\label{corollary6}
			When $ \al+\beta=1 $ and $ \beta=\gamma \ge 0 $, the SE of the $i$th user pair as $ M \to \infty $ is written as
			\begin{align}
				R_{i}=\min\left(R_{i}^\text{MAC},R_{i}^\text{BC}\right),
			\end{align}
			where 
			\begin{align}
				&\!\!\!R_{i}=\frac{\tau_{\mathrm{c}}-\tau_{\mathrm{p}}}{2 \tau_{\mathrm{c}}}\nn\\
				&\!\!\!\!\times\!\log_{2}\!\!\left(\!\!1+\frac{\displaystyle NE_{\mathrm{u}}\!\!\left(\!\!\left(\sum_{m=1}^{M}\phi_{\mathrm{A},mi}\right)^{\!\!2}+\left(\sum_{m=1}^{M}\phi_{\mathrm{B},mi}\right)^{\!\!2}\right)}{\sum_{m=1}^{M} \left(\phi_{\mathrm{A},mi}+\phi_{\mathrm{B},mi}\right)}\!\!\right)\!\!\\
				&R_{i}^\text{BC}= \min\left(R_{\mathrm{A},i}^\text{MAC}, R_{\mathrm{B},i}^\text{BC}\right) +\min\left(R_{\mathrm{B},i}^\text{MAC}, R_{\mathrm{A},i}^\text{BC}\right),
			\end{align}
			with $ R_{\mathrm{X},i}^\text{MAC} $ and $ R_{\mathrm{X},i}^\text{BC} $ given in terms of \eqref{MACB32} and \eqref{MACB33}, respectively.
			
		\end{corollary} 
		
		In other words, both conditions $ \al+\beta=1 $ and $ \al+\gamma=1 $ are fulfilled. By shedding further light on this corollary, we observe an interplay appearing among the transmit powers. For example, a reduction of the pilot transmit power would result in a degradation of the estimated channel that could be balanced by an increase of the transmit power of the users/APs during the transmission phase to preserve the performance with respect to the SE.
		
		\begin{remark}
			In all corollaries above, the corresponding SE could be boosted by increasing the involved $ E_{\mathrm{i}} $, $ i\mathrm{=p, u, r} $. For example, Corollary \ref{corollary2} suggests that $ R_{i} $ can be increased with $ E_{\mathrm{r}} $ by increasing the transmit power of the APs.
		\end{remark}
		
		\section{Power Allocation}\label{PA}
		Different from the previous section, where the transmit powers of all users were assumed equal for the sake of exposition of the scaling laws, in this section, we elaborate on the optimal power allocation among the users and the APs during both MAC and BC phases, respectively. We assume that the power allocation takes place during the data transmission phase while the power design of the training phase in terms of $ p_{\mathrm{p}} $ has previously being determined. In particular, we follow the procedure in \cite{Ngo2014,Kong2018} and adapt it according to our system architecture having $ M $ distributed APs as relay nodes.
		
		We focus on the maximization of the sum SE constrained to a total power $ P $, i.e., $ p_{\mathrm{u}} \sum_{i=1}^{W} \left(\eta_{\mathrm{A},i}+\eta_{\mathrm{B},i}\right)+p_{\mathrm{r}}\le P$. In particular, the formulation of the power allocation optimization is described by
		\begin{subequations}
			\begin{alignat}{2}
				&\max_{\etav_{\mathrm{A}},\etav_{\mathrm{B}},p_{\mathrm{r}}}&\quad&\sum_{i=1}^{W} {R}_{i}\label{opt1}\\
				&\mathrm{subject}~\mathrm{to}& & p_{\mathrm{u}} \sum_{i=1}^{W} \left(\eta_{\mathrm{A},i}+\eta_{\mathrm{B},i}\right)+p_{\mathrm{r}}\le P\nn\\
				& &&\etav_{\mathrm{A}}\ge\b0, \etav_{\mathrm{B}}\ge\b0, p_{\mathrm{d}}\ge 0, p_{\mathrm{r}}\ge 0\nn\\
				& &&{R}_{i}\ge R_{\mathrm{min}}, i \in \mathcal{W}\nn
			\end{alignat}
		\end{subequations}
		where we have denoted $ \etav_{\mathrm{A}}=\left[\etav_{\mathrm{A},1},\ldots, \etav_{\mathrm{A},W}\right]^{\T} $ and $ \etav_{\mathrm{B}}=\left[\etav_{\mathrm{B},1},\ldots, \etav_{\mathrm{B},W}\right]^{\T} $, while $ R_{\mathrm{min}} $ expresses the minimum SE of the $ i $th pair. Given that the logarithm is an increasing function, the optimization, given by \eqref{opt1}, can be written as
		\begin{subequations}
			\begin{alignat}{2}
				&\min_{\substack{\tilde{\etav}_{\mathrm{A}},\tilde{\etav}_{\mathrm{B}},p_{\mathrm{r}}\\ \gamma_{i}, \gamma_{\mathrm{A},i}, \gamma_{\mathrm{B},i}
				}}&\quad&\prod_{i=1}^{W}\left(1+{\gamma}_{i}\right)^{-1}\label{opt2}\\
				&\mathrm{subject}~\mathrm{to}& &{\gamma}_{i}\le \frac{a_{1,i}\tilde{\eta}_{\mathrm{A},i}+a_{2,i}\tilde{\eta}_{\mathrm{B},i}}{\sum_{j=1}^{W}\left(a_{3,ij}\tilde{\eta}_{\mathrm{A},j}+a_{4,ij}\tilde{\eta}_{\mathrm{B},j}\right)+1}\label{opt21}\\
				& && \gamma_{\mathrm{A},i}\le \min\{\frac{a_{1,i}\tilde{\eta}_{\mathrm{A},i}}{c_{i}}, \frac{p_{\mathrm{r}}}{p_{\mathrm{r}}b_{B,i}+c_{B,i}}\}, i \in \mathcal{W}\\
				& && \gamma_{\mathrm{B},i}\le \min\{\frac{a_{2,i}\tilde{\eta}_{\mathrm{B},i}}{c_{i}}, \frac{p_{\mathrm{r}}}{p_{\mathrm{r}}b_{A,i}+c_{A,i}}\}, i \in \mathcal{W}\\
				& &&\gamma_{i}\le \gamma_{\mathrm{A},i}+ \gamma_{\mathrm{B},i}+\gamma_{\mathrm{A},i}\gamma_{\mathrm{B},i}, i \in \mathcal{W}\label{opt22}\\
				& &&\sum_{i=1}^{W} \left(\tilde{\eta}_{\mathrm{A},i}+\tilde{\eta}_{\mathrm{B},i}\right)+p_{\mathrm{r}}\le P\\
				& &&\tilde{\etav}_{\mathrm{A}}\ge\b0, \tilde{\etav}_{\mathrm{B}}\ge\b0, p_{\mathrm{d}}\ge 0, p_{\mathrm{r}}\ge 0\\
				& &&\gamma_{i}^{-1}\left(2^{\frac{2 \tau_{\mathrm{c}}R_{\mathrm{min}}}{\tau_{\mathrm{c}}-\tau_{\mathrm{p}}}}-1\right)\le 1, i \in \mathcal{W}
			\end{alignat}
		\end{subequations}
		where we have defined $\tilde{ \eta}_{\mathrm{A},i}=p_{\mathrm{u}}\eta_{\mathrm{A},i} $, $\tilde{ \eta}_{\mathrm{B},i}=p_{\mathrm{u}}\eta_{\mathrm{B},i} $, $a_{1,i}=\frac{N \left(\sum_{m=1}^{M}\phi_{\mathrm{A},mi}\right)^{2}}{\sum_{m=1}^{M}\left(\phi_{\mathrm{A},mi}+\phi_{\mathrm{B},mi}\right)} $, $a_{2,i}=\frac{N \left(\sum_{m=1}^{M}\phi_{\mathrm{B},mi}\right)^{2}}{\sum_{m=1}^{M}\left(\phi_{\mathrm{A},mi}+\phi_{\mathrm{B},mi}\right)} $, $a_{3,ij}=\frac{\sum_{m=1}^{M}\al_{\mathrm{A},mj}}{\sum_{m=1}^{M}\left(\phi_{\mathrm{A},mi}+\phi_{\mathrm{B},mi}\right)} $, $a_{4,ij}=\frac{\sum_{m=1}^{M}\al_{\mathrm{B},mj}}{\sum_{m=1}^{M}\left(\phi_{\mathrm{A},mi}+\phi_{\mathrm{B},mi}\right)} $, $c_{i}= \sum_{j=1}^{W}\left(a_{3,ij}\tilde{\eta}_{\mathrm{A},j}+a_{4,ij}\tilde{\eta}_{\mathrm{B},j}\right)+1 $, $ b_{X,i}=\frac{\sum_{m=1}^{M}\sum_{j=1}^{W}\al_{\mathrm{X},mi}\left(\eta_{\mathrm{A},mj}\phi_{\mathrm{B},mj}+\eta_{\mathrm{B},mj}\phi_{\mathrm{A},mj}\right)}{N\left(\sum_{m=1}^{M}\sqrt{\eta_{\mathrm{\bar{X}},mi}}\phi_{\mathrm{X},mi}\right)^{2}} $, and $ c_{X,i}=\frac{\sum_{m=1}^{M}\sum_{j=1}^{W}\left(\eta_{\mathrm{A},mj}\phi_{\mathrm{B},mj}+\eta_{\mathrm{B},mj}\phi_{\mathrm{A},mj}\right)}{N\left(\sum_{m=1}^{M}\sqrt{\eta_{\mathrm{\bar{X}},mi}}\phi_{\mathrm{X},mi}\right)^{2}} $. Moreover, $ \gamma_{i} $, $ \gamma_{\mathrm{A},i} $, and $ \gamma_{\mathrm{B},i} $ correspond to the SINRs of $ 	R_{i} $, $ \min\left(R_{\mathrm{A},i}^\text{MAC}, R_{\mathrm{B},i}^\text{BC}\right) $ and $ \min\left(R_{\mathrm{B},i}^\text{MAC}, R_{\mathrm{A},i}^\text{BC}\right) $, respectively.
		
		The latter optimization problem is nonconvex since it falls to the category of complementary geometric programming (CGP). However, its solution can be obtained by solving a sequence of convex GP problems \cite{Chiang2007,Weeraddana2011,Ngo2014,Kong2018}. Initially, as in \cite[Lem. 1]{Weeraddana2011}, we approximate the objective function $ \left(1+{\gamma}_{i}\right)$ by the monomial function $ \delta_{i}{\gamma}_{i}^{\mu_{i}} $, where $ \delta_{i}=\left({\gamma}_{i}^{\circ}\right)^{-\mu_{i}} \left(1+{\gamma}_{i}^{\circ}\right)$ and $ \mu_{i}=\frac{{\gamma}_{i}^{\circ}}{{\gamma}_{i}^{\circ}+1} $ \cite{Ngo2014}. Next, we transform the two inequalities described by \eqref{opt21} and \eqref{opt22} into posynomials, in order to result in a GP problem. Specifically, the first inequality is written as
		\begin{align}
			\!\!a_{1,i}\tilde{\eta}_{\mathrm{A},i}\!+\!a_{2,i}\tilde{\eta}_{\mathrm{B},i}\! \ge\! \left(
			\frac{a_{1,i}\tilde{\eta}_{\mathrm{A},i}}{\phi_{\mathrm{A},i}}\right)^{\!\phi_{\mathrm{A},i}}\!\!\left(\frac{a_{2,i}\tilde{\eta}_{\mathrm{B},i}}{\phi_{\mathrm{B},i}}\right)^{\!\phi_{\mathrm{B},i}}, \label{firstIneq}
		\end{align}
		where $ \phi_{\mathrm{A},i}=\frac{a_{1,i}\tilde{\eta}^{\circ}_{\mathrm{A},i}}{a_{1,i}\tilde{\eta}^{\circ}_{\mathrm{A},i}+a_{2,i}\tilde{\eta}^{\circ}_{\mathrm{B},i} } $, $ \phi_{\mathrm{B},i}=\frac{a_{2,i}\tilde{\eta}^{\circ}_{\mathrm{B},i}}{a_{1,i}\tilde{\eta}^{\circ}_{\mathrm{A},i}+a_{2,i}\tilde{\eta}^{\circ}_{\mathrm{B},i} } $ with $ \tilde{\eta}^{\circ}_{\mathrm{A},i} $ and $ \tilde{\eta}^{\circ}_{\mathrm{B},i} $ denoting the initialization values. In \eqref{firstIneq}, we have applied a known property, expressing that, for any set of positive numbers, the geometric mean is no larger than the arithmetic mean \cite{He2014}. Hence, substitution of \eqref{firstIneq} into \eqref{opt21} gives
		\begin{align}
			{\gamma}_{i}\le \frac{\left(
				\frac{a_{1,i}\tilde{\eta}_{\mathrm{A},i}}{q_{\mathrm{A},i}}\right)^{q_{\mathrm{A},i}}\left(\frac{a_{2,i}\tilde{\eta}_{\mathrm{B},i}}{q_{\mathrm{B},i}}\right)^{q_{\mathrm{B},i}}}{\sum_{j=1}^{W}\left(a_{3,ij}\tilde{\eta}_{\mathrm{A},j}+a_{4,ij}\tilde{\eta}_{\mathrm{B},j}\right)+1}, i \in \mathcal{W}.
		\end{align}
		
		For the second inequality, given by \eqref{opt22}, we follow the procedure in \cite{Kong2018} to approximate $ f\left(x,y\right)=x+y+xy $ near an arbitrary point $ x^{\circ}, y^{\circ}>0 $ by means of the monomial function $ g\left(x,y\right) = \zeta x^{\lambda_{1}}y^{\lambda_{2}}$ in terms of $ \zeta $, $ \lambda_{1} $, and $ \lambda_{2} $. These parameters are obtained in \cite{Kong2018} as $ \lambda_{1}=\frac{x^{\circ}\left(1+y^{\circ}\right)}{x^{\circ}+y^{\circ}+x^{\circ}y^{\circ}} $, $ \lambda_{2}=\frac{y^{\circ}\left(1+x^{\circ}\right)}{x^{\circ}+y^{\circ}+x^{\circ}y^{\circ}} $, $ \zeta=\left(x^{\circ}+y^{\circ}+x^{\circ}y^{\circ}\right)\left(x^{\circ}\right)^{-\lambda_{1}} \left(y^{\circ}\right)^{-\lambda_{2}}$. Thus, we have
		\begin{align}
			\gamma_{i}\le \zeta_{i} \gamma_{\mathrm{A},i}^{\lambda_{\mathrm{A},i}}\gamma_{\mathrm{B},i}^{\lambda_{\mathrm{B},i}}, i \in \mathcal{W}
		\end{align}
		where $ \zeta_{i}=\left(\gamma_{\mathrm{A},i}^{\circ}+\gamma_{\mathrm{B},i}^{\circ}+\gamma_{\mathrm{A},i}^{\circ}\gamma_{\mathrm{B},i}^{\circ}\right)\left(\gamma_{\mathrm{A},i}^{\circ}\right)^{-\lambda_{\mathrm{A},i}}\left(\gamma_{\mathrm{B},i}^{\circ}\right)^{-\lambda_{\mathrm{B},i}}$, $ \lambda_{\mathrm{A},i}=\frac{\gamma_{\mathrm{A},i}^{\circ}\left(1+\gamma_{\mathrm{B},i}^{\circ}\right)}{\gamma_{\mathrm{A},i}^{\circ}+\gamma_{\mathrm{B},i}^{\circ}+\gamma_{\mathrm{A},i}^{\circ}\gamma_{\mathrm{B},i}^{\circ}} $, and $ \lambda_{\mathrm{B},i}=\frac{\gamma_{\mathrm{B},i}^{\circ}\left(1+\gamma_{\mathrm{A},i}^{\circ}\right)}{\gamma_{\mathrm{A},i}^{\circ}+\gamma_{\mathrm{B},i}^{\circ}+\gamma_{\mathrm{A},i}^{\circ}\gamma_{\mathrm{B},i}^{\circ}} $ with $ \gamma_{\mathrm{A},i}^{\circ} $, $ \gamma_{\mathrm{B},i}^{\circ} $ being the initialization values. The algorithm steps are provided by Algorithm 1, where the parameter $ \theta>1 $ defines the desired accuracy as a trade-off with convergence speed. Especially, as $ \theta $ approaches $ 1 $, we result in better accuracy while the convergence speed is slow.
		\begin{algorithm}\label{Algoa1}
			\caption{Successive approximation algorithm}
			1.				 \textit{Initialisation}: Define the parameter $ \theta $ and the tolerance $ \epsilon $. Set $ k=1 $, $\tilde{\eta}_{\mathrm{A},i}=\tilde{\eta}_{\mathrm{B},i}=\frac{P}{4 W} $, $ p_{\mathrm{r}}=\frac{P}{2} $ while $ \gamma_{i}^{\circ} $, $ \gamma_{\mathrm{A},i}^{\circ} $, and $ \gamma_{\mathrm{B},i}^{\circ} $ are chosen by means of Theorem~\ref{theoremTotalSE}.\\
			2. Iteration $ k $: Evaluate $ \mu_{i}=\frac{\gamma_{i}^{\circ}}{\gamma_{i}^{\circ}+1} $, $\phi_{\mathrm{A},i}=\frac{a_{1,i}\tilde{\eta}^{\circ}_{\mathrm{A},i}}{a_{1,i}\tilde{\eta}^{\circ}_{\mathrm{A},i}+a_{2,i}\tilde{\eta}^{\circ}_{\mathrm{B},i} } $, $ \phi_{\mathrm{B},i}=\frac{a_{2,i}\tilde{\eta}^{\circ}_{\mathrm{B},i}}{a_{1,i}\tilde{\eta}^{\circ}_{\mathrm{A},i}+a_{2,i}\tilde{\eta}^{\circ}_{\mathrm{B},i} } $, $ \zeta_{i}=\left(\gamma_{\mathrm{A},i}^{\circ}+\gamma_{\mathrm{B},i}^{\circ}+\gamma_{\mathrm{A},i}^{\circ}\gamma_{\mathrm{B},i}^{\circ}\right)\left(\gamma_{\mathrm{A},i}^{\circ}\right)^{-\lambda_{\mathrm{A},i}}\left(\gamma_{\mathrm{B},i}^{\circ}\right)^{-\lambda_{\mathrm{B},i}}$, $ \lambda_{\mathrm{A},i}=\frac{\gamma_{\mathrm{A},i}^{\circ}\left(1+\gamma_{\mathrm{B},i}^{\circ}\right)}{\gamma_{\mathrm{A},i}^{\circ}+\gamma_{\mathrm{B},i}^{\circ}+\gamma_{\mathrm{A},i}^{\circ}\gamma_{\mathrm{B},i}^{\circ}} $, $ \lambda_{\mathrm{B},i}=\frac{\gamma_{\mathrm{B},i}^{\circ}\left(1+\gamma_{\mathrm{A},i}^{\circ}\right)}{\gamma_{\mathrm{A},i}^{\circ}+\gamma_{\mathrm{B},i}^{\circ}+\gamma_{\mathrm{A},i}^{\circ}\gamma_{\mathrm{B},i}^{\circ}} $. Then, solve the GP problem:

			\begin{subequations}
				\begin{alignat}{2}
					&\min_{\substack{\tilde{\etav}_{\mathrm{A}},\tilde{\etav}_{\mathrm{B}},p_{\mathrm{r}}\\ \gamma_{i}, \gamma_{\mathrm{A},i}, \gamma_{\mathrm{B},i}
					}}&\quad&\prod_{i=1}^{W}{\gamma}_{i}^{-\mu_{i}}\label{opt3}\\
					&\mathrm{subject}~\mathrm{to}& & \theta^{-1}\tilde{\eta}_{\mathrm{A},i}\le\tilde{\eta}_{\mathrm{A},i}\le \theta\tilde{\eta}_{\mathrm{A},i}, i \in \mathcal{W} \\
					& && \theta^{-1}\tilde{\eta}_{\mathrm{B},i}\le\tilde{\eta}_{\mathrm{B},i}\le \theta\tilde{\eta}_{\mathrm{B},i}, i \in \mathcal{W}\\
					& && \theta^{-1}\gamma_{i}^{\circ}\le \gamma_{i}\le \theta\gamma_{i}^{\circ}, i \in \mathcal{W}\\
					& && \theta^{-1}\gamma_{\mathrm{A},i}^{\circ}\le \gamma_{\mathrm{A},i}\le \theta\gamma_{\mathrm{A},i}^{\circ}, i \in \mathcal{W}\\
					& && \theta^{-1}\gamma_{\mathrm{B},i}^{\circ}\le \gamma_{\mathrm{B},i}\le \theta\gamma_{\mathrm{B},i}^{\circ}, i \in \mathcal{W}\\
					& &&	c_{i}{\gamma}_{i}\!\left(\!
					\frac{a_{1,i}\tilde{\eta}_{\mathrm{A},i}}{\phi_{\mathrm{A},i}}\!\right)^{\!-\phi_{\mathrm{A},i}}\!\!\left(\!\frac{a_{2,i}\tilde{\eta}_{\mathrm{B},i}}{\phi_{\mathrm{B},i}}\!\right)^{\!\!-\phi_{\mathrm{B},i}} \!\!\le 1, i \in \mathcal{W}\\
					& &&\gamma_{i}\zeta_{i}^{-1} \gamma_{\mathrm{A},i}^{-\lambda_{\mathrm{A},i}}\gamma_{\mathrm{B},i}^{-\lambda_{\mathrm{B},i}}\le 1, i \in \mathcal{W}\\
					& && 			 \gamma_{\mathrm{A},i}c_{i} \left(a_{1,i}\tilde{\eta}_{\mathrm{A},i}\right)^{-1} \le 1 i \in \mathcal{W}\\
					& && 			 \gamma_{\mathrm{B},i}c_{i} \left(a_{2,i}\tilde{\eta}_{\mathrm{B},i}\right)^{-1} \le 1 i \in \mathcal{W}\\
					& && \gamma_{\mathrm{A},i}p_{\mathrm{r}}^{-1} \left(p_{\mathrm{r}}b_{B,i}+c_{B,i}\right)\le 1, i \in \mathcal{W}\\
					& && \gamma_{\mathrm{B},i}p_{\mathrm{r}}^{-1} \left(p_{\mathrm{r}}b_{A,i}+c_{A,i}\right)\le 1, i \in \mathcal{W}\\
					& &&\sum_{i=1}^{W} \left(\tilde{\eta}_{\mathrm{A},i}+\tilde{\eta}_{\mathrm{B},i}\right)+p_{\mathrm{r}}\le P\\
					& &&\tilde{\etav}_{\mathrm{A}}\ge\b0, \tilde{\etav}_{\mathrm{B}}\ge\b0,p_{\mathrm{d}}\ge 0, p_{\mathrm{r}}\ge 0\\
					& &&\gamma_{i}^{-1}\left(2^{\frac{2 \tau_{\mathrm{c}}R_{\mathrm{min}}}{\tau_{\mathrm{c}}-\tau_{\mathrm{p}}}}-1\right)\le 1, i \in \mathcal{W}
				\end{alignat}
			\end{subequations}
			
			Let $ \tilde{\eta}_{\mathrm{A},i}^{\star} $, $ \tilde{\eta}_{\mathrm{B},i}^{\star} $, $ \gamma_{i}^{\star} $, $ \gamma_{\mathrm{A},i}^{\star} $, $ \gamma_{\mathrm{B},i}^{\star} $, $ i \in \mathcal{W} $.\\
			3. If $ \max_{i}|\tilde{\eta}_{\mathrm{A},i}^{\star}-\tilde{\eta}_{\mathrm{A},i}^{\circ}|<\epsilon$ and/or $ \max_{i}|\tilde{\eta}_{\mathrm{B},i}^{\star}-\tilde{\eta}_{\mathrm{B},i}^{\circ}|<\epsilon$ and/or $ \max_{i}| \gamma_{i}^{\star}- \gamma_{i}^{\circ}|<\epsilon$ and/or $ \max_{i}| \gamma_{\mathrm{A},i}^{\star}- \gamma_{\mathrm{A},i}^{\circ}|<\epsilon$ and/or $ \max_{i}| \gamma_{\mathrm{B},i}^{\star}- \gamma_{\mathrm{B},i}^{\circ}|<\epsilon$ $ \rightarrow $ Stop. Otherwise, go to step $ 4 $.\\
			4. \textit{Update initial values}. Set $ k= k +1 $, $ \tilde{\eta}_{\mathrm{A},i}^{\circ}=\tilde{\eta}_{\mathrm{A},i}^{\star} $, $ \tilde{\eta}_{\mathrm{B},i}^{\circ}=\tilde{\eta}_{\mathrm{B},i}^{\star} $, $ \gamma_{i}^{\circ}=\gamma_{i}^{\star}$, $ \gamma_{\mathrm{A},i}^{\circ}=\gamma_{\mathrm{A},i}^{\star}$, $ \gamma_{\mathrm{B},i}^{\circ}=\gamma_{\mathrm{B},i}^{\star}$, and go to step $ 2 $.
		\end{algorithm}

		\section{Numerical Results}\label{Numerical} 
		This section depicts the analytical results provided by means of Theorem~\ref{theoremTotalSE} and Propositions~\ref{PropositionPEproof}-\ref{PropositionPhaseC} that illustrate the performance of a multi-pair two-way CF mMIMO system. For the sake of comparison, we have accounted for a conventional two-way collocated massive MIMO architecture employing DF as described by \cite{Kong2018}. Also, our analytical results are accompanied by Monte Carlo simulations by means of $ 10^{3} $ independent channel realizations, in order to verify them and show their tightness. 
		
		\begin{table}
			\caption{Parameters Values for Numerical Results~}
			\begin{center}
				\begin{tabulary}{\columnwidth}{ | c | c | }\hline
					{\bf Description} &{\bf Values}\\ \hline
					Number of APs& $M=200$\\ \hline
					Number of Antennas/AP & $N=3$\\ \hline
					Number user pairs& $W=5$\\ \hline
					Carrier frequency & $f_{0} = 2~\mathrm{GHz}$\\ \hline
					Power per pilot symbol & $\bar{p}_{\mathrm{p}}=100~\mathrm{mW}$\\ \hline
					Uplink transmit power &
					${\bar{p}_{\mathrm{u}}}=100~\mathrm{mW}$ \\ \hline
					Path loss exponent & $\al=4$ \\ \hline
					Communication bandwidth & $W_{\mathrm{c}} = 20~\mathrm{MHz}$\\ \hline		
					Coherence bandwidth & $B_{\mathrm{c}}=200~\mathrm{KHz}$ 
					\\ \hline
					Coherence time & $T_{\mathrm{c}}=1~\mathrm{ms}$
					\\ \hline
					Duration of uplink training & $\tau_{\mathrm{p}}=10$ samples
					\\ \hline
					Boltzmann constant & $\kappa_{\mathrm{B}}=1.381\times 10^{-23}~\mathrm{J/K}$ \\ \hline
					Noise temperature & $T_{0}= 290~\mathrm{K}$ \\ \hline
					Noise figure & ${N_\mathrm{F}}=9~\mathrm{dB}$ \\ \hline
				\end{tabulary}\label{ParameterValues1} 
			\end{center}
		\end{table}

		\subsection{Simulation Setup}
		Unless otherwise stated, the following
		set of parameters is used during the simulations. In particular, we consider $ M=200 $ APs and $ W=5 $ user pairs uniformly distributed in an area of $ 1 \times 1~\mathrm{km}^{2}$. Each AP is equipped with $ N=3 $ antennas. Note that the area edges are wrapped around to avoid the boundary effects. Also, we assume that the coherence time and bandwidth are $T_{\mathrm{c}}=1~\mathrm{ms}$ and $B_{\mathrm{c}}=200~\mathrm{kHz}$, respectively, which means that the coherence block consists of $ 200 $ channel uses. The orthogonality among pilots requires at least $ \tau_{\mathrm{p}}=2W=10 $. Moreover, we assume that ${p}_{\mathrm{p}}$, ${{p}_{\mathrm{u}}}$, and ${{p}_{\mathrm{r}}}$ correspond to the normalized powers, obtained by dividing $\bar{p}_{\mathrm{p}}=\bar{p}_{\mathrm{u}}$, and $\bar{p}_{\mathrm{r}}$ by the noise power ${N_\mathrm{P}} = W_{\mathrm{c}} $ $\mbox{\footnotesize{$\times$}}$ $ \kappa_{\mathrm{B}}$ $\mbox{\footnotesize{$\times$}}$ $ T_{0}$ $\mbox{\footnotesize{$\times$}}$ ${N_\mathrm{F}}$.
		The various parameters are found in Table~\ref{ParameterValues1}. Power control in terms of maximizing the sum SE is considered only in Subsection \ref{Powerallocation}. Hence, without any power control, we assume that in the uplink, all users transmit with full power, i.e., $ \eta_{\mathrm{A},i}=\eta_{\mathrm{B},i}=1, \forall i$. Similarly, in the downlink, all APs transmit with full power, which means $ \eta_{mi}=\left(N\sum_{i=1}^{W}\left( \phi_{\mathrm{B},mi}+ \phi_{\mathrm{A},mi}\right) \right)^{-1}$ by satisfying \eqref{prPower1}. Also, the transmit powers during the MAC and BC phases are assumed equal, i.e., $ 2 W {p}_{\mathrm{u}} ={p}_{\mathrm{r}}$.

		We take into account for \cite[Remark 4]{Bjoernson2020}, and thus, we consider the 3GPP Urban Microcell model in~\cite[Table B.1.2.1-1]{3GPP2017} as a more appropriate benchmark for CF mMIMO systems than the established model presented initially in~\cite{Ngo2017} because of two main reasons: i) although CF mMIMO systems are more likely suggested for shorter distances, the model in~\cite{Ngo2017} assumes that shadowing is met for users found $ 50\mathrm{m} $ further from an AP, and ii) the COST-Hata model, used in~\cite{Ngo2017} is suitable for macro-cells with APs being at least
		$ 1\mathrm{km} $ far from the users and at least $ 30\mathrm{m} $ above the ground while the CF setting suggests APs found at lower height and being very close to the users. Specifically, the large-scale fading coefficient, described by this mode for a $ 2~\mathrm{GHz}$ carrier frequency, is given by
		\begin{align}
			\al_{\mathrm{X},mk}[dB]=-30.5-36.7\log_{10}\left(\frac{d_{\mathrm{X},mk}}{1~\mathrm{m}}\right)+F_{\mathrm{X},mk},\nn
		\end{align}
		where $ d_{\mathrm{X},mk} $ expresses the distance between AP $ m $ and user $k $ while $ F_{\mathrm{X},mk}\sim\mathcal{CN}\left(0,4^{2}\right) $ describes the shadow fading. In addition, the shadowing terms between different users are assumed to be correlated as $\EE\{F_{\mathrm{X},mk}F_{\mathrm{X},ij}\}=4^{2}2^{-\delta_{\mathrm{X},kj}/9} $ only when $ m=i $, where $ \delta_{\mathrm{X},kj} $ is the distance between users $ k $ and $j$. 
		
		\subsection{Demonstration of basic properties}
		Initially, we assume that $ p_{\mathrm{p}}= p_{\mathrm{u}} $ as well as $ p_{\mathrm{r}}=2 W p_{\mathrm{u}} $, which mean that users transmit equal power during the training and data transmission phases. 
		
		Fig. \ref{Fig2} presents the sum SE versus $ p_{\mathrm{u}} $ with varying number of APs along with Monte-Carlo simulations verifying the analytical expressions since the lines almost coincide. Also, we notice that SE saturates at high SNR due to the inter-user interference, as expected. Moreover, SE increases with the increasing number of APs $ M $.

		Fig \ref{Fig1} depicts the sum SE versus the number of APs $ M $. Also, we have considered the scenario of genie receivers at the users during the BC phase. In other words, we have assumed that the corresponding receivers are aware of instantaneous CSI and not just its statistics. Since the gap between the two lines is small, the downlink channel hardens and no extra training is required. In addition, for the sake of comparison, we have included the collocated scenario with a base station at the center of the area being the relay node and having $ M N $ antennas, $ \al_{\mathrm{X},mk}=\al_{\mathrm{X}k} $, $ \eta_{k}=M \eta_{mk} $, $ \forall m $, and in general, all the parameters equal across the index $ m $ \cite{Vu2019}. It can be seen that the CF mMIMO relay setting outperforms the collocated layout because the diversity against path-loss and shadow fading is exploited. Moreover, we have considered the conventional orthogonal scheme where the transmission of each pair takes place at different time slots or frequency bands. As the number of APs increases, the two-way CF MIMO performs better because the effect from the inter-user interference decreases. Thus, when $ M $ is low, the orthogonal scheme performs better, but the mMIMO system behaves better as $ M $ increases and channels become orthogonal, which means a mitigation of the interference. Compared to Fig. \ref{Fig1}, which assumes $ W=5 $ user pairs, Fig. \ref{Fig11} shows the sum SE in the case of $ W=20$ user pairs. In the latter figure, the outperformance of the CF mMIMO setting over the collocated scenario is more pronounced because of its concomitant advantages. Hence, we observe that the performance gap at $ M=400 $ AP is $ 8\% $ and $ 39\% $, when $ W=5 $ and $ W=20 $, respectively. 		
		\subsection{Power-scaling laws}
		Herein, we verify Propositions \ref{PropositionPEproof}-\ref{PropositionPhaseC} providing the power-scaling laws, also denoted as asymptotic results that correspond to the scenarios A-C mentioned earlier. In addition, we elaborate on the resultant power savings with comparison to the analytical exact results provided by Theorem \ref{theoremTotalSE}.
		
		\begin{figure}[!h]
			\begin{center}
				\includegraphics[width=0.9\linewidth]{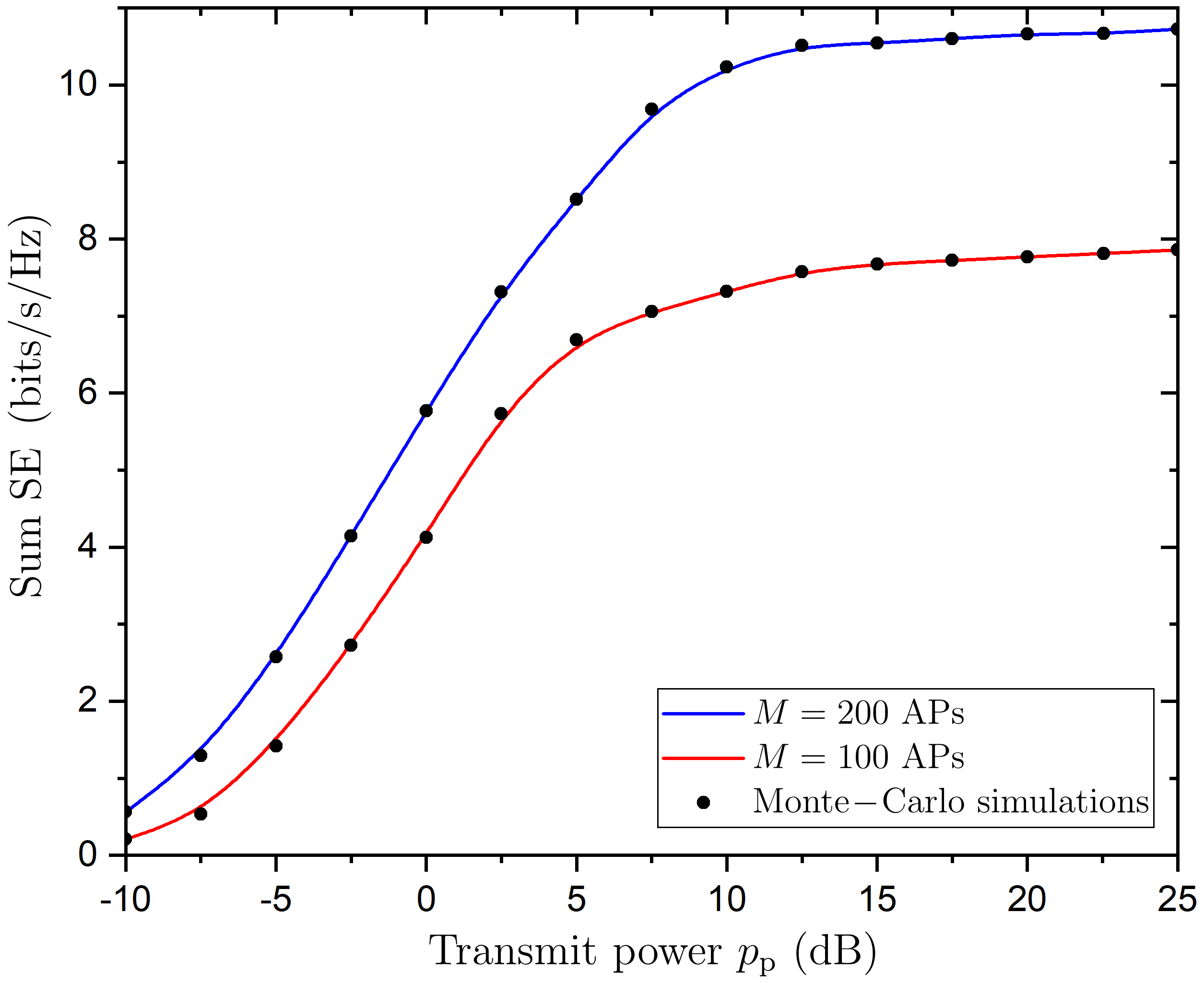}
				\caption{\footnotesize{Sum SE per versus the uplink transmit power $ p_{\mathrm{p}} =p_{\mathrm{u}} $ for varying number of APs $M$ with validation by Monte-Carlo simulations ($ N=3 $,  $ W=5 $, and $ p_{\mathrm{r}}=2 W p_{\mathrm{u}} $).}}
				\label{Fig2}
			\end{center}
		\end{figure}
		\begin{figure}[!h]
			\subfigure[]{
				\centering
				\includegraphics[width=0.9\linewidth]{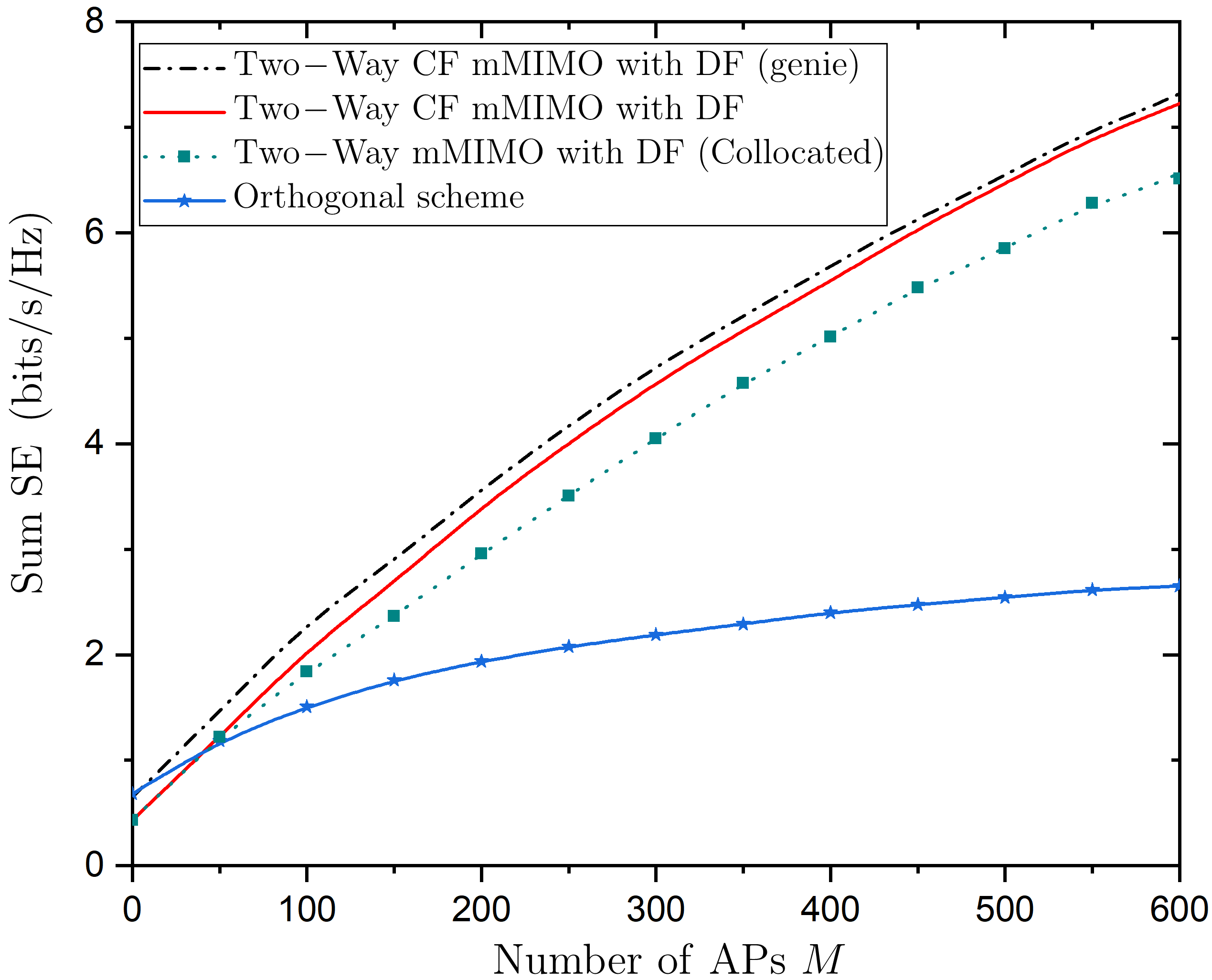}	\label{Fig1}}
			
			\subfigure[]{
				\centering
				\includegraphics[width=0.9\linewidth]{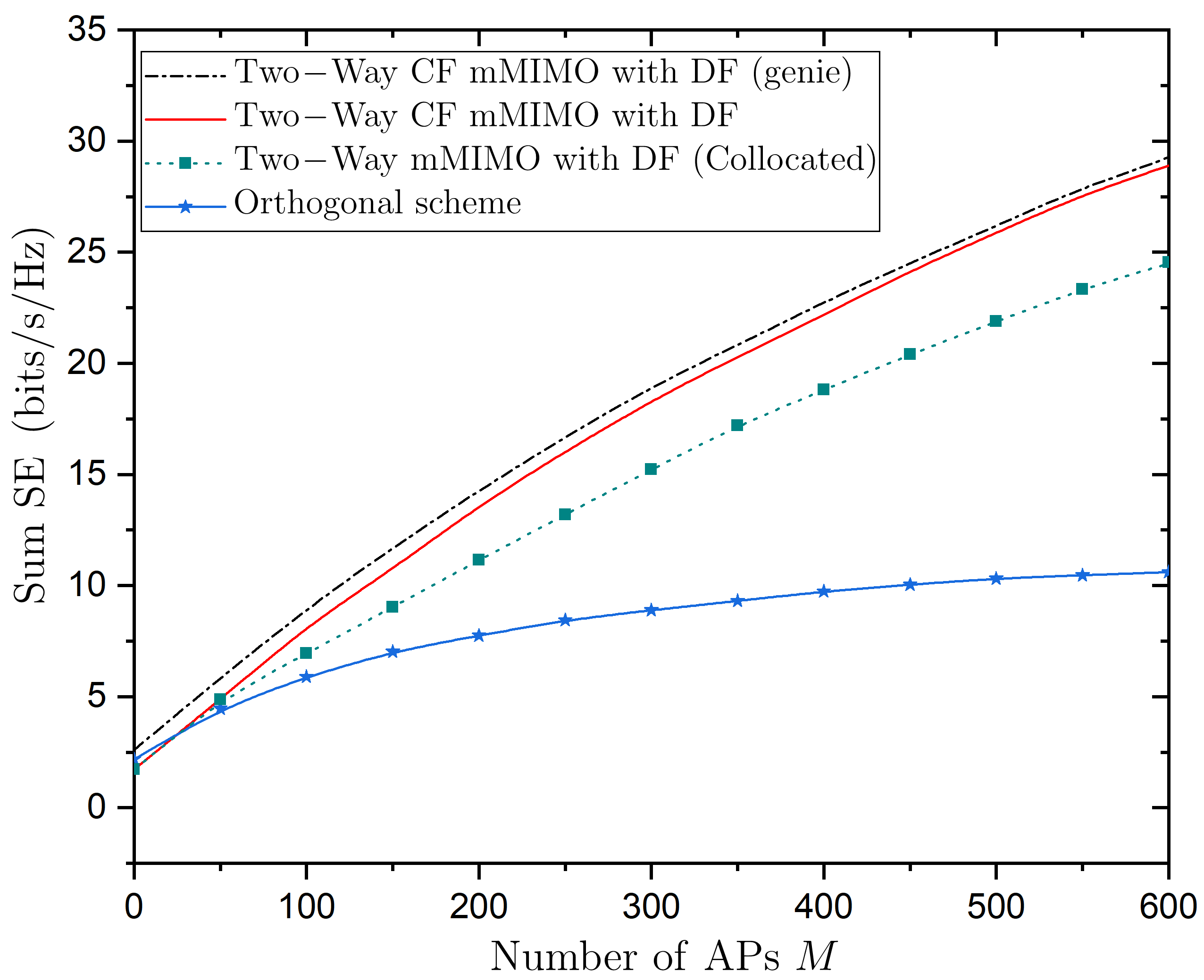}	\label{Fig11}}
			\caption{\footnotesize{Sum SE versus the number of APs $ M $ for different scenarios ($ N=3 $, $ p_{\mathrm{p}}= p_{\mathrm{u}} $, and $ p_{\mathrm{r}}=2 W p_{\mathrm{u}} $), when (a) $ W=5 $ and (b) $ W=20 $ user pairs.}}
			
		\end{figure}
		
		Fig~\ref{Fig3} sheds light into Scenario A by depicting the sum SE with respect to the number of APs $ M $ for varying scaling in terms of the parameter $ \al $. In general, we observe that the asymptotic results approach the exact curves as $ M \to \infty $. When $ \al=0.7 \left(0< \al<1\right)$, $ R_{i} \to \infty$, while if $ \al=1 $, the asymptotic SE saturates and approaches the analytical result. The third group of curves corresponds to $ \al=1.4>1 $. In such case, $ R_{i} $ tends to zero.

		Figs. \ref{Fig4} and \ref{Fig5} illustrate the properties regarding the power savings of Scenario B described by Proposition \ref{PropositionPhaseB}. Especially, in any of the cases i) $ \beta=1 $ and $ 0<\gamma<1 $, ii) $ 0<\beta<1$ and $ \gamma=1 $, and iii) $ \beta=\gamma=1 $, we show in Fig. \ref{Fig4} that the asymptotic results converge to specific values and approach the exact results in the large number of APs regime ($ M \to \infty $) as described by Corollaries 3, 4. In the upper set of plots of Fig. \ref{Fig5}, we observe that when $ \beta $ or $ \gamma $ is greater than one, which means that the transmit power of the MAC or BC phase is cut down too much, the sum SE approaches zero as $ M \to \infty $. In fact, the larger the parameter, being greater than one, the faster the decrease of the SE to zero. Furthermore, when both the transmit powers of users and APs are reduced tolerably such that $ \beta<1$ and $ \gamma<1 $, the sum SE $ R_{i} $ increases without bound as the lower set of lines of Fig. \ref{Fig5} reveals.
		\begin{figure}[!h]
			\begin{center}
				\includegraphics[width=0.95\linewidth]{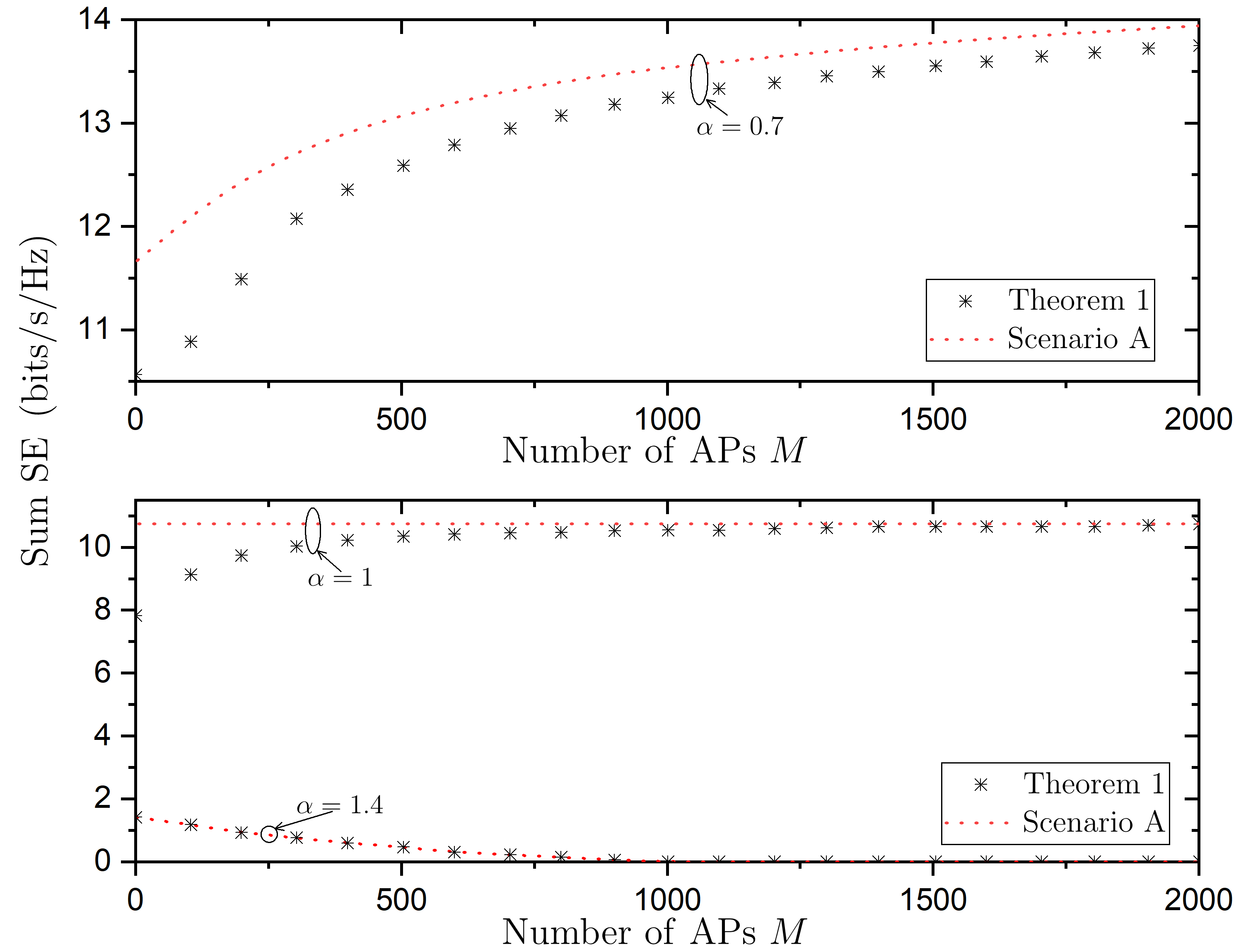}
				\caption{\footnotesize{Sum SE versus the number of APs $M$ by means of asymptotic (Scenario A) and exact analysis (Theorem \ref{theoremTotalSE}) for $ N=3 $,  $ W=5 $, $ p_{\mathrm{p}}= {E_{\mathrm{p}}}/{M^{\al} }$ with $ E_{\mathrm{p}}= 10~\mathrm{dB} $.}}
				\label{Fig3}
			\end{center}
		\end{figure}
		
		Fig. \ref{Fig6} represents Scenario C describing the interplay between the pilot symbol power and and the users/relay powers. The sums $ \al+\beta $ and $ \al+\gamma $ determine the behavior of the sum SE in the large number of APs limit. Hence, if we set $ \al=1.1 $, $ \beta= 1.2 $, $ \gamma= 0.4$ and $ \al=0.9 $, $ \beta= 1.4 $, $ \gamma= 0.6$, we observe that both lines converge to zero as $ M \to \infty $ since, in both cases, we have $ \al+\beta =2.3$ and $ \al+\gamma=1.5$. Notably, the two lines converge to each other (their gap decreases) as $ M $ increases, i.e., the asymptotic sum SE is the same because the two sums are kept identical. Furthermore, the line with $ \al=0.9 $ provides better SE for finite number of APs although the transmit relay power is cut down more because the channel is estimated with higher quality. The middle set of lines demonstrates that the sum SE grows without bound when both $ \al+\beta <1$ and $ \al+\gamma<1$ simultaneously. The third subfigure illustrates that if any of the two following conditions are satisfied, the sum SE approaches a non-zero limit. Specifically, if i) $ \al+\beta =1$ and $ \beta> \gamma\ge 0 $ or ii) $ \al+\gamma=1 $ and $ \gamma> \beta \ge 0 $, then $ R_{i} $ saturates.
		
		\begin{figure}[!h]
			\begin{center}
				\includegraphics[width=0.95\linewidth]{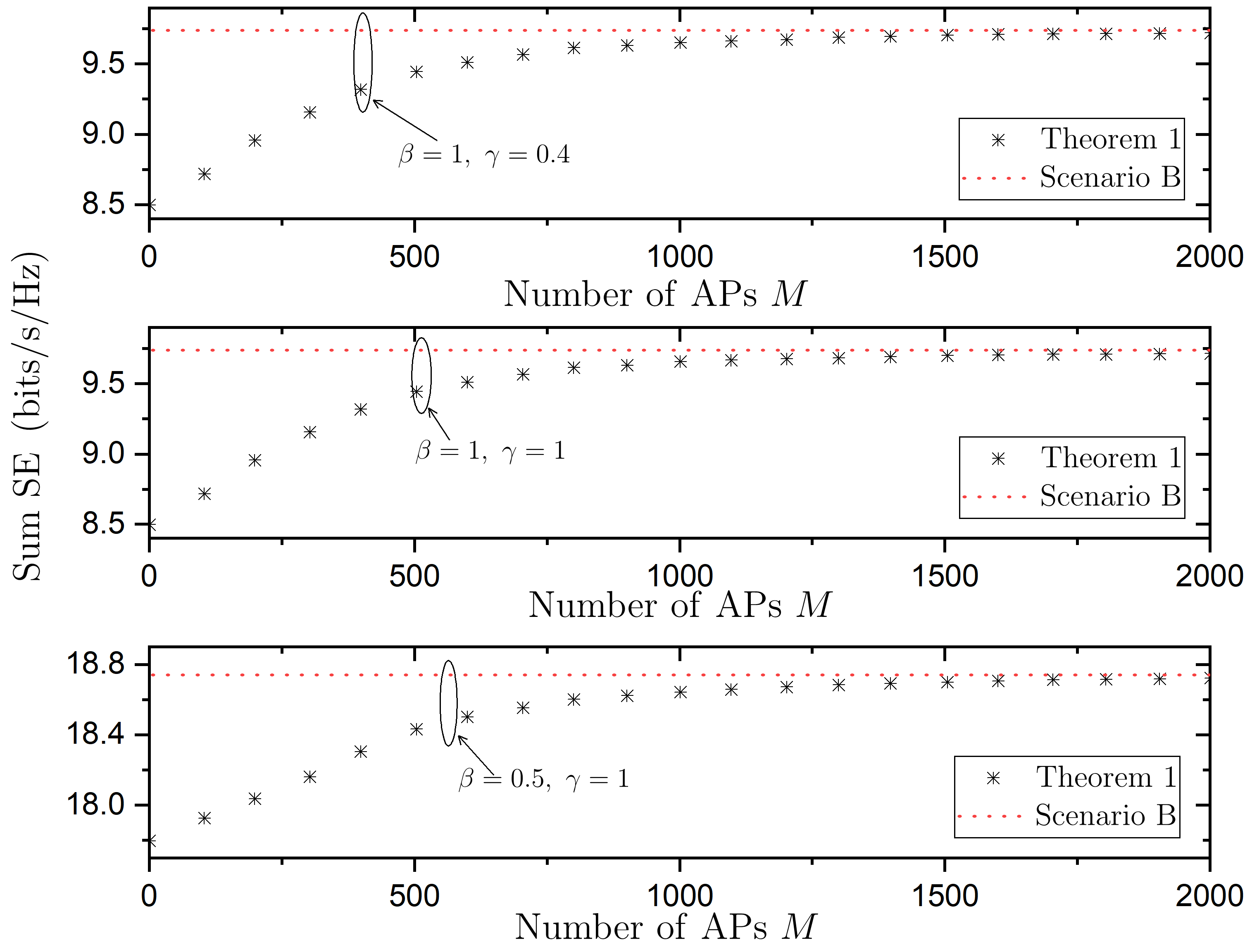}
				\caption{\footnotesize{Sum SE versus the number of APs $M$ by means of asymptotic (Scenario B) and exact analysis (Theorem \ref{theoremTotalSE}) for $ N=3 $, $ W=5 $, $ p_{\mathrm{u}}= {E_{\mathrm{u}}}/{M^{\beta} }$ with $ E_{\mathrm{u}}= 10~\mathrm{dB} $, and $ p_{\mathrm{r}}= {E_{\mathrm{r}}}/{M^{\gamma} }$ with $ E_{\mathrm{r}}= 10~\mathrm{dB} $ (finite limits).}}
				\label{Fig4}
			\end{center}
		\end{figure}
		\begin{figure}[!h]
			\begin{center}
				\includegraphics[width=0.95\linewidth]{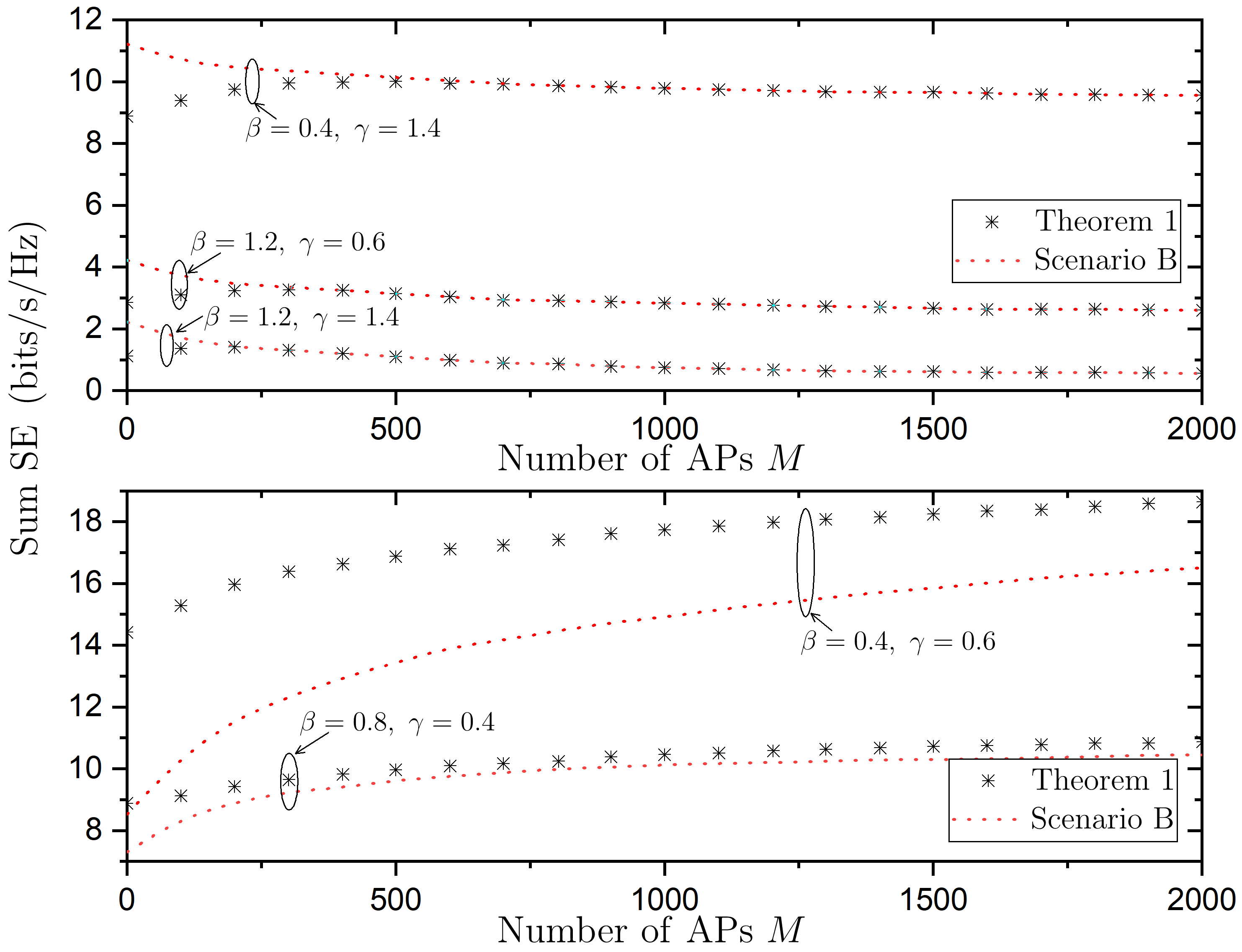}
				\caption{\footnotesize{Sum SE versus the number of APs $M$ by means of asymptotic (Scenario B) and exact analysis (Theorem \ref{theoremTotalSE}) for $ N=3 $,  $ W=5 $, $ p_{\mathrm{u}}= {E_{\mathrm{u}}}/{M^{\beta} }$ with $ E_{\mathrm{u}}= 10~\mathrm{dB} $, and $ p_{\mathrm{r}}= {E_{\mathrm{r}}}/{M^{\gamma} }$ with $ E_{\mathrm{r}}= 10~\mathrm{dB} $ (zero and unbounded limits).}}
				\label{Fig5}
			\end{center}
		\end{figure}
		\begin{figure}[!h]
			\begin{center}
				\includegraphics[width=0.95\linewidth]{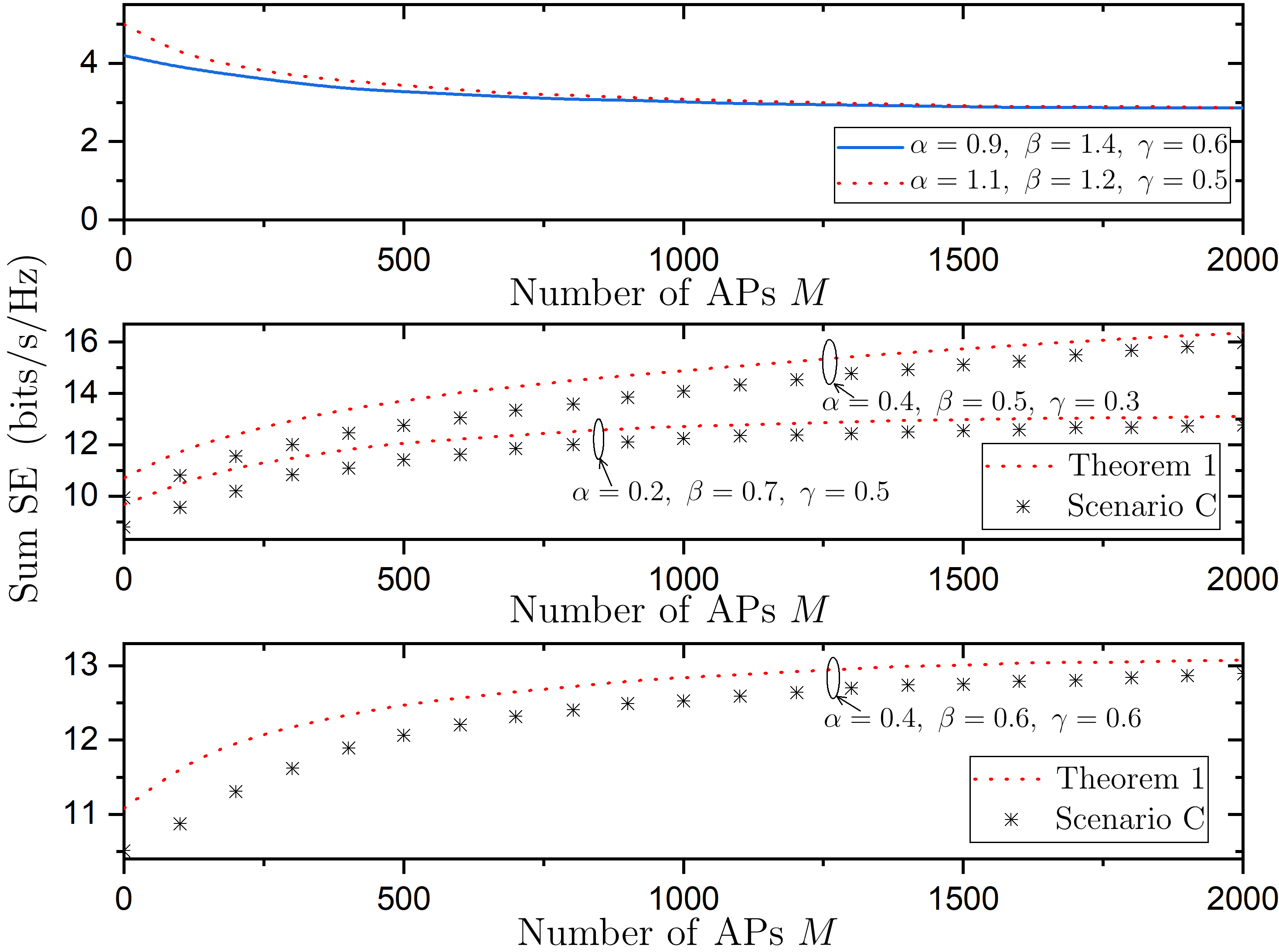}
				\caption{\footnotesize{Sum SE versus the number of APs $M$ by means of asymptotic (Scenario C) and exact analysis (Theorem \ref{theoremTotalSE}) for $ N=3 $, $ W=5 $, $ p_{\mathrm{p}}= {E_{\mathrm{p}}}/{M^{\al} }$ with $ E_{\mathrm{p}}= 10~\mathrm{dB} $, $ p_{\mathrm{u}}= {E_{\mathrm{u}}}/{M^{\beta} }$ with $ E_{\mathrm{u}}= 10~\mathrm{dB} $, and $ p_{\mathrm{r}}= {E_{\mathrm{r}}}/{M^{\gamma} }$ with $ E_{\mathrm{r}}= 10~\mathrm{dB} $ (zero, unbounded, and finite limits).}}
				\label{Fig6}
			\end{center}
		\end{figure}
		\begin{figure}[!h]
			\begin{center}
				\includegraphics[width=0.9\linewidth]{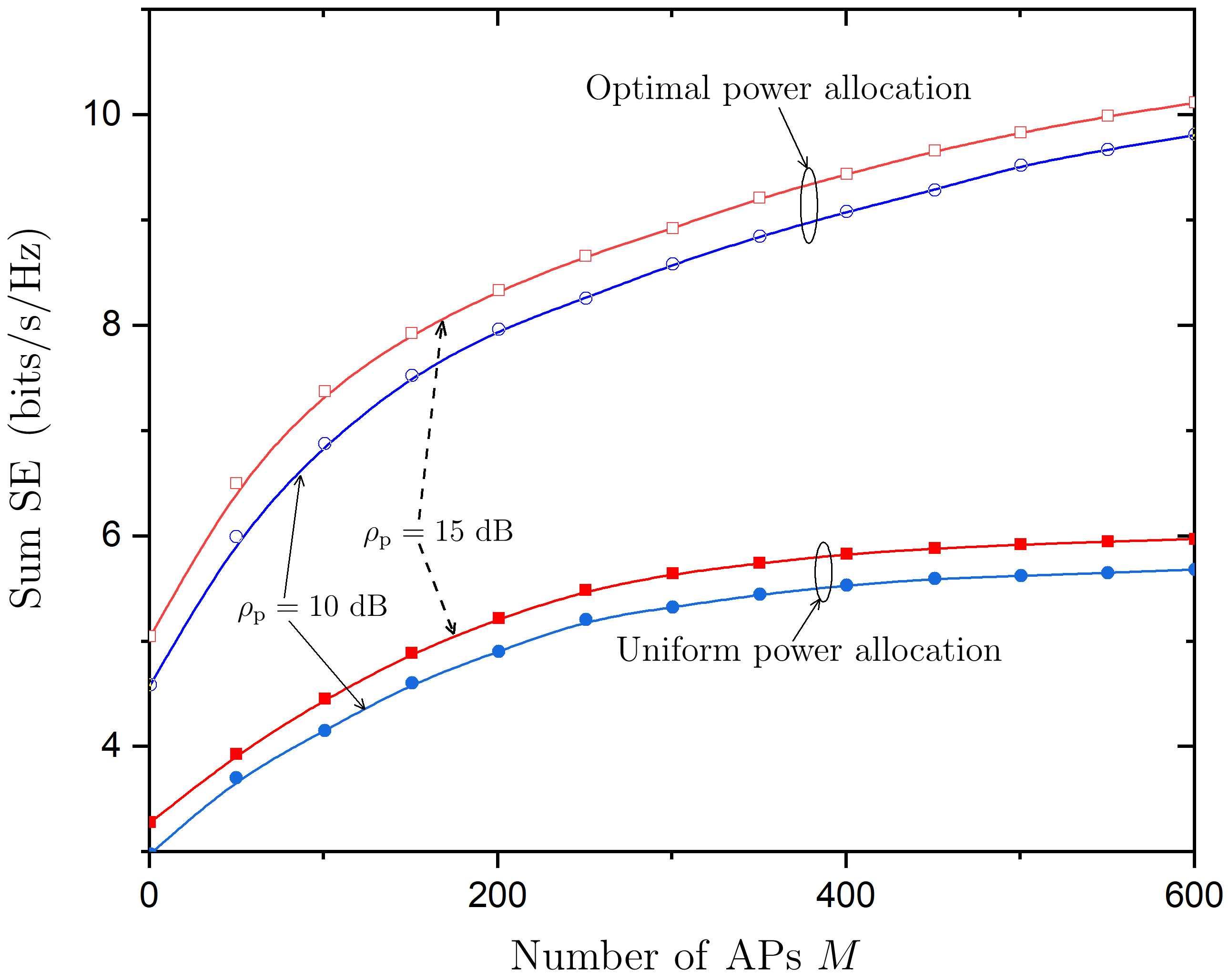}
				\caption{\footnotesize{Sum SE versus the number of APs $M$ by means of uniform (Theorem 1) and optimal (Algorithm 1) power allocations for $ p_{\mathrm{p}}= 10~\mathrm{dB} $ (blue lines) and $ p_{\mathrm{p}}= 15~\mathrm{dB} $ (red lines).}}
				\label{Fig7}
			\end{center}
		\end{figure}

		\subsection{Power allocation}\label{Powerallocation}
		Fig. \ref{Fig7} shows the performance of the two-way CF mMIMO system with optimal power allocation using Algorithm 1 when $ M=200 $ APs and with varying channel estimation accuracy by means of $ p_{\mathrm{p}} $ while $ \theta=1.1 $ to achieve good accuracy with reasonable convergence time. Also, we assume that the total power budget is $ P=10~\mathrm{dB} $. Moreover, we consider uniform power allocation by using Theorem \ref{theoremTotalSE} for the sake of comparison. Notably, the optimal power allocation performs better than uniform power allocation, and it results in an improvement of $ 63.2\% $ and $ 59.3\% $ when $ p_{\mathrm{p}} =10~dB$ and $ p_{\mathrm{p}} =15~dB$, respectively. In other words, a better channel estimation barely affects the sum SE.
		\section{Conclusion}\label{Conclusion}
		This paper investigated the sum SE of a multipair two-way HD relaying system assisted by a CF mMIMO architecture employing MR processing and accounting for imperfect CSI. Contrary to the common collocated mMIMO layout, the distributed CF mMIMO design achieves higher performance. Furthermore, power-scaling laws, achieving to scale the transmit powers of the users and APs while maintaining the desired SE, were obtained. Also, the trade-offs regarding these laws was examined. Finally, we performed an optimal power allocation concerning the transmit powers of the APs and users during the data transmission phase towards the improvement of the SE with comparison to uniform power allocation.
		\begin{appendices}
			\section{Proof of Theorem~\ref{theoremTotalSE}}\label{TotalSEproof}
			
			%
			We start with the derivation of $ \gamma_{i}^\text{MAC} $. The desired signals of $ \mathrm{T}_{\mathrm{A},i} $ and $ \mathrm{T}_{\mathrm{B},i} $ in \eqref{gammaMac} are written as
			\begin{align}
				\EE\Big\{\sum_{m=1}^{M}\!\!\left(\!\hatvh^{\H}_{mi}\bh_{mi}+\hatvg^{\H}_{mi}\bh_{mi}\!\right)\!\!\!\Big\}&\!=\!N
				\!\!\sum_{m=1}^{M}\!\!\phi_{\mathrm{A},mi},\label{MAC1}\\
				\EE\Big\{\sum_{m=1}^{M}\!\!\left(\!\hatvh^{\H}_{mi}\bg_{mi}+\hatvg^{\H}_{mi}\bg_{mi}\!\right)\!\!\!\Big\}&\!=\!N
				\!\!\sum_{m=1}^{M}\!\!\phi_{\mathrm{B},mi},\label{MAC2}
			\end{align}
			since $ \bh_{mi} $ and $ \bg_{mi} $ are independent.
			We continue with the derivations of $ \mathrm{EE}_{\mathrm{A},i}^\text{MAC} $, $ \mathrm{EE}_{\mathrm{B},i}^\text{MAC} $, $ \mathrm{IUI}_{i}^\text{MAC} $, and $ \mathrm{N}_{i}^\text{MAC} $. Specifically, we have
			\begin{align}
				\mathrm{EE}_{\mathrm{A},i}^\text{MAC}&=\eta_{\mathrm{A},i}\!\left(\!\var\Big\{\sum_{m=1}^{M}\hatvh^{\H}_{mi}\bh_{mi}\Big\}+\var\Big\{\sum_{m=1}^{M}\hatvg^{\H}_{mi}\bh_{mi}\Big\}\!\!\right)\!,\label{MAC31}
			\end{align}
			where \eqref{MAC31} is obtained because the variance of a sum of independent RVs is equal to the sum of the variances.
			The first term of \eqref{MAC31} is obtained as
			\begin{align}
				&\var\Big\{\sum_{m=1}^{M}\hatvh^{\H}_{mi}\bh_{mi}\Big\}=	\sum_{m=1}^{M} \EE\Big\{\Big|\hat{\bh}_{mi}^{\H}{\bh}_{mi}-\EE\Big\{\hat{\bh}_{mi}^{\H}{\bh}_{mi}\Big|^{2}\Big\}\Big\}\label{MAC34}\\
				&=\sum_{m=1}^{M} \left( \EE\Big\{\Big|\hat{\bh}_{mi}^{\H}{\bh}_{mi}\Big|^{2}\Big\}-\Big|\EE\{\hat{\bh}_{mi}^{\H}\hat{\bh}_{mi}\}\Big|^{2}\right)\nn\\
				&=\sum_{m=1}^{M} \left( \EE\Big\{\Big|\|\hat{\bh}_{mi}\|^{2}+\hat{\bh}_{mi}^{\H}\tilde{\bh}_{mi}\Big|^{2}\Big\}-N^{2}\phi_{\mathrm{A},mi}^{2}\right)\label{MAC39}\\
				&=\!\sum_{m=1}^{M}\!\! \left(\EE\Big\{\|\hat{\bh}_{mi}\|^{4}\Big\}\!+\! \EE\Big\{|\hat{\bh}_{mi}^{\H}\tilde{\bh}_{mi}|^{2}\Big\}\!-\!N^{2}\phi_{\mathrm{A},mi}^{2}\right)\label{MAC35}\\
				&=\!N\!\sum_{m=1}^{M}\!\! \left(\left(N+1\right)\phi^{2}_{\mathrm{A},mi}\!+\!\phi_{\mathrm{A},mi}{e}_{\mathrm{A},mi}\!-\!N^{2}\phi_{\mathrm{A},mi}^{2}\right)\label{MAC36}\\
				&=N\sum_{m=1}^{M} \al_{\mathrm{A},mi}\phi_{\mathrm{A},mi},\label{MAC38}
			\end{align}
			where \eqref{MAC34} follows again because the variance of a sum of independent RVs is equal to the sum of the variances. In \eqref{MAC39}, we have considered that $ \tilde{\bh}_{mi} $ is independent of $\hatvh_{mi} $ and has zero mean. The identity $\EE\{\|\hat{\bh}_{mi}\|^{4}\}=N\left(N+1 \right) \phi^{2}_{\mathrm{A},mi}$ has been used in \ref{MAC36}, and \eqref{MAC38} follows after some algebraic manipulations since $ e_{\mathrm{A},mi}=\al_{\mathrm{A},mi}- \phi_{\mathrm{A},mi} $.
			More easily, the second term of \eqref{MAC31} is given by
			\begin{align}
				\var\Big\{\sum_{m=1}^{M}\hatvg^{\H}_{mi}\bh_{mi}\Big\}=	N\sum_{m=1}^{M} \al_{\mathrm{A},mi}\phi_{\mathrm{B},mi}\label{MAC381}.
			\end{align}	
			Hence, $ 	\mathrm{EE}_{\mathrm{A},i}^\text{MAC} $ becomes my means of \eqref{MAC38} and \eqref{MAC381}
			\begin{align}
				\mathrm{EE}_{\mathrm{A},i}^\text{MAC}=N\eta_{\mathrm{A},i}\sum_{m=1}^{M} \al_{\mathrm{A},mi}\left(\phi_{\mathrm{A},mi}+\phi_{\mathrm{B},mi}\right)\label{MAC382}.
			\end{align}
			In the same way, we derive 
			\begin{align}
				\mathrm{EE}_{\mathrm{B},i}^\text{MAC}=N\eta_{\mathrm{B},i}\sum_{m=1}^{M} \al_{\mathrm{B},mi}\left(\phi_{\mathrm{A},mi}+\phi_{\mathrm{B},mi}\right)\label{MAC383}.
			\end{align}
			The term, describing the inter-user interference, is obtained as
			\begin{align}
				&\mathrm{IUI}_{i}^\text{MAC}\!=\!\sum_{j\ne i}^{W}\!\!\left(\! \eta_{\mathrm{A},j}
				\sum_{m=1}^{M}\!\!\left(\EE\{|\hatvh^{\H}_{mi}\bh_{mj}|^{2} \}+\EE\{|\hatvg^{\H}_{mi}\bh_{mj}|^{2}\}\!\right)\!\!\right)\nn\\
				&+ \sum_{j\ne i}^{W}\!\! \left(\eta_{\mathrm{B},j} \sum_{m=1}^{M}\left(\EE\{|\hatvh^{\H}_{mi}\bg_{mj}|^{2}\}+
				\EE\{|\sum_{m=1}^{M}\hatvg^{\H}_{mi}\bg_{mj}|^{2}\}\!\right)\!\!\right)\nn\\
				&=\sum_{j\ne i}^{W}\left( \eta_{\mathrm{A},j}
				\sum_{m=1}^{M}\left(\phi_{\mathrm{A},mi}\al_{\mathrm{A},mj}+\phi_{\mathrm{B},mi}\al_{\mathrm{B},mj}\right)\right)\nn\\
				&+\sum_{j\ne i}^{W} \left(\eta_{\mathrm{B},j} \sum_{m=1}^{M}\left(\phi_{\mathrm{A},mi}\al_{\mathrm{B},mj}+
				\phi_{\mathrm{B},mi}\al_{\mathrm{A},mi}\right)\right)\nn\\
				&=\!N\!\sum_{j\ne i}^{W}\!\sum_{m=1}^{M}\!\!\left( \eta_{\mathrm{A},j}\al_{\mathrm{A},mj}\!+\! \eta_{\mathrm{B},j}\al_{\mathrm{B},mj}\right)\!\! \left(\phi_{\mathrm{A},mi}\!+\!\phi_{\mathrm{B},mi}\right)\!,\!\!\!\label{MAC384}
			\end{align}
			where we have applied the property $ \EE\{|X+Y|^{2}\}=\EE\{|X|^{2}\}+\EE\{|Y|^{2}\} $, holding between two independent random variables $ X $ and $ Y $ with $ \EE\{X\}=0 $. Next, the noise term becomes
			\begin{align}
				\mathrm{N}_{i}^\text{MAC}&=\frac{1}{p_{\mathrm{u}}}\sum_{m=1}^{M}\left( \phi_{\mathrm{A},mi}+ \phi_{\mathrm{B},mi}\right).\label{MAC385}
			\end{align}
			Substitution of \eqref{MAC1}, \eqref{MAC2},\eqref{MAC382}, \eqref{MAC383}, \eqref{MAC384}, and \eqref{MAC385} into \eqref{gammaMac} and \eqref{received8} provides $ R_{i}^\text{MAC} $ as well as $ R_{\mathrm{X},i}^\text{MAC} $ for $ \mathrm{X}\in\{\mathrm{A},\mathrm{B}\} $.\\
			The proof continues with the derivation of $ R_{\mathrm{A},i}^\text{BC} $ by means of the computation of \eqref{BC1}-\eqref{BC5}. Specifically, regarding the desired signal, we have
			\begin{align}
				&\EE\Big\{\sum_{m=1}^{M}\!\!\sqrt{p_{\mathrm{d},m}\eta_{\mathrm{B},mi}}\bh_{mi}^{\T}\hat{\bh}_{mi}^{*}\Big\}\!=\!\EE\Big\{\!\sum_{m=1}^{M}\!\!\sqrt{p_{\mathrm{d},m}\eta_{\mathrm{B},mi}}\|\hat{\bh}_{mi}\|^{2}\Big\}\nn\\
				&\!+\!\EE\Big\{\sum_{m=1}^{M}\sqrt{p_{\mathrm{d},m}\eta_{\mathrm{B},mi}}\tilde{\bh}_{mi}^{\T}\hat{\bh}_{mi}^{*}\Big\}\nn\\
				&\!=\!N\sum_{m=1}^{M}\sqrt{p_{\mathrm{d},m}\eta_{\mathrm{B},mi}}\phi_{\mathrm{A},mi},\label{BC7}
			\end{align}
			where we have used the independence between $ \hatvh_{mi} $ and $ \tilde{\bh}_{mi} $.
			Also, $ \mathrm{BU}_{A_{i}}^\text{BC} $ and $ \mathrm{BU}_{B_{i}}^\text{BC} $ are obtained by following the same procedure with \eqref{MAC38}. In particular, we have
			\begin{align}
				\mathrm{BU}_{A_{i}}^\text{BC}
				&=N \sum_{m=1}^{M}\eta_{\mathrm{B},mi}\al_{\mathrm{A},mi}\phi_{\mathrm{A},mi},\label{BC8}\\
				\mathrm{BU}_{B_{i}}^\text{BC}&=N \sum_{m=1}^{M}\eta_{\mathrm{A},mi}\al_{\mathrm{A},mi}\phi_{\mathrm{B},mi}.\label{BC9}
			\end{align}
			Next, regarding $ \mathrm{IUI}_{A_{i}}^\text{BC} $, we have
			\begin{align}
				& \EE\Big\{\Big|\sum_{m=1}^{M}\eta_{\mathrm{A},mj}^{1/2}\bh_{mi}^{\T}\hat{\bg}_{mj}^{*}\Big|^{2}\Big\}= \sum_{m=1}^{M}\eta_{\mathrm{A},mj}\EE\{|\bh_{mi}^{\T}\hat{\bg}_{mj}^{*}|^{2}\}\nn\\
				&=N \sum_{m=1}^{M}\eta_{\mathrm{A},mj}\al_{\mathrm{A},mi}\phi_{\mathrm{B},mj},\label{BC10}
			\end{align}
			where in the first equation, we have used again that $ \EE\{|X+Y|^{2}\}=\EE\{|X|^{2}\}+\EE\{|Y|^{2}\} $.
			The last term, $ \mathrm{IUI}_{B_{i}}^\text{BC} $, is obtained similarly as
			\begin{align}
				\mathrm{IUI}_{B_{i}}^\text{BC} =N \sum_{m=1}^{M}\eta_{\mathrm{B},mj}\al_{\mathrm{A},mi}\phi_{\mathrm{A},mj}.\label{BC11}
			\end{align}
			By using \eqref{BC7}-\eqref{BC11} and \eqref{prPower}, we obtain $ R_{\mathrm{A},i}^\text{BC} $ and conclude the proof since $ R_{\mathrm{B},i}^\text{BC} $ can be derived in the same fashion.

			
			
		\end{appendices}
		\bibliographystyle{IEEEtran}
		
		\bibliography{mybib}
		
		
	\end{document}